\documentclass[a4paper,UKenglish,cleveref, autoref, thm-restate]{lipics-v2021} 
\usepackage{graphicx} 
\usepackage{svg}
\usepackage{hyperref}
\usepackage{tabularx}
\usepackage{booktabs}
\usepackage{amsmath}
\usepackage{amssymb}
\usepackage{amsthm}
\usepackage{blkarray}
\usepackage{mathtools}
\usepackage{algorithm} 
\usepackage{algpseudocode} 

\usepackage{etoolbox} 
\usepackage{apptools}
\AtAppendix{\counterwithin{theorem}{section}}
\AtAppendix{\counterwithin{lemma}{section}}
\AtAppendix{\counterwithin{proposition}{section}}
\AtAppendix{\counterwithin{corollary}{section}}
\AtAppendix{\counterwithin{definition}{section}}
\AtAppendix{\counterwithin{figure}{section}}

\usepackage{thmtools}

\usepackage{pgf}
\usepackage{tikz}
\usetikzlibrary{arrows,automata}
\usetikzlibrary{positioning}
\usetikzlibrary{calc}
\usetikzlibrary{intersections}
\usetikzlibrary{shapes,fit} 
\usetikzlibrary{math}

\usepackage{float}

\let\textcite\cite 

\usepackage[small,full]{complexity}
\newclass{\UEOPL}{UEOPL}
\newlang{\wSAT}{Weighted\;2\text{-}SAT}
\newlang{\VCSP}{VCSP}

\newcommand{\height}[1][]{\text{height}_f(#1)} 
\newcommand{\width}[1][]{\text{width}_f(#1)} 

\newcommand{\arc}[3][r]{#2\drct[#1]#3} 

\newcommand{\drct}[1][r]{
    \if#1r \rightarrow{}\fi
    \if#1l \leftarrow{} \fi
    \if#1b \leftrightarrow{}\fi
} 

\newcommand{\se}[3][]{#2\xrightarrow{#1}#3}
\newcommand{\nse}[3][]{#2\not\xrightarrow{#1}#3}
\newcommand{\rse}[3][]{#2\xleftrightarrow{#1}#3}
\newcommand{\nrse}[3][]{#2\not\xleftrightarrow{#1}#3}

\DeclareMathOperator*{\argmax}{arg\,max}

\newcommand{\appsymb}{$\star$}
\newcommand{\appref}[1]{{\hyperref[proof:#1]{\appsymb}}}  

\title{When is local search both effective and efficient?}

\author{Artem Kaznatcheev}{Department of Mathematics, and Department of Information and Computing Sciences, Utrecht University, The  Netherlands}{a.kaznatcheev@uu.nl}{https://orcid.org/0000-0001-8063-2187}{}

\author{Sofia {Vazquez Alferez}}{Department of Mathematics, and Department of Information and Computing Sciences, Utrecht University, The Netherlands}{s.vazquezalferez@uu.nl}{https://orcid.org/0000-0002-1541-8683}{}

\authorrunning{A. Kaznatcheev and S. Vazquez Alferez} 

\Copyright{Artem Kaznatcheev and Sofia Vazquez Alferez}



\nolinenumbers

\hideLIPIcs

\begin{document}

\makeatletter
\let\c@author\relax
\makeatother

\makeatletter
\def\thmt@innercounters{equation,algorithm}
\makeatother

\maketitle

\begin{abstract}
Combinatorial optimization problems implicitly define \emph{fitness landscapes} that combine the numeric structure of the `fitness' function to be maximized with the combinatorial structure of which assignments are `adjacent'.
Local search starts at an assignment in this landscape and successively moves assignments until no further improvement is possible among the adjacent assignments.
Classic analyses of local search algorithms have focused more on the question of effectiveness (``did we find a \textit{good} solution?'') and often implicitly assumed that there are no doubts about their efficiency (``did we find it \textit{quickly}?'').
But there are many reasons to doubt the efficiency of local search.
Even if we focus on fitness landscapes on the hypercube that are single peaked on every subcube (known as semismooth fitness landscapes, completely unimodal pseudo-Boolean functions, or acyclic unique sink orientations) where effectiveness is obvious,
many local search algorithms are known to be inefficient.
Since fitness landscapes are unwieldy exponentially large objects, we focus on their polynomial-sized representations by instances of valued constraint satisfaction problems (\VCSP).
We define a ``direction'' for valued constraints such that \emph{directed} \VCSP{s} generate semismooth fitness landscapes.
We call directed \VCSP{s} \emph{oriented} if they do not have any pair of variables with arcs in both directions.
Since recognizing if a \VCSP-instance is directed or oriented is \coNP-complete, we generalized oriented \VCSP{s} as \emph{conditionally-smooth} fitness landscapes where the structural property of `conditionally-smooth' \emph{is} recognizable in polynomial time for a \VCSP-instance.
We prove that many popular local search algorithms like random ascent, simulated annealing, history-based rules, jumping rules, and the Kernighan-Lin heuristic are very efficient on conditionally-smooth landscapes.
But conditionally-smooth landscapes are still expressive enough so that other well-regarded local search algorithms like steepest ascent and random facet require a super-polynomial number of steps to find the fitness peak.
\end{abstract}

\newpage
\section{Introduction}
\label{sec:intro}

Local search algorithms start at an initial assignment and successively move to better adjacent assignments until no further improvement is possible.
As with all algorithms, we can ask about their effectiveness (``did we find a \textit{good} solution?'') and efficiency (``did we find it 
\textit{quickly}?'').
Often, we do not know the precise answers to these questions, but still choose to use local search in combinatorial optimisation~\cite{LocalSearch_Book1,LocalSearch_Book4,LocalSearch_Book3,LocalSearch_Book2,LocalSearchCO_Thesis}.\footnote{
Sometimes local search is chosen because it is easy to implement.
If we do not expect any algorithm to be effective and efficient on all inputs -- say for \NP-hard combinatorial optimization problems -- then why not choose an algorithm that will take less effort to implement and maintain?
At other times, local search is forced on us.
If we want to understand the evolution of biological systems~\cite{evoPLS,KazThesis} or of social systems in business~\cite{businessLandscape,orgBehavOrig} and economics~\cite{AlgoEconSurvey}, or even just energy-minimization in physical systems, then we need to study the effectiveness and efficiency of the local search algorithms followed by nature.
}
In practice, local search seems to efficiently find effective solutions.
But practice is not theory.

As theorists, we can view combinatorial optimization problems as implicitly defining \emph{fitness landscapes} that combine the numeric structure of the `fitness' function to be maximized with the combinatorial structure of which assignments are `adjacent'~\cite{LocalSearch_Book1,LocalSearch_Book4,repCP}.
Since many local search algorithms stop at local peaks in these fitness landscapes,
much work asks questions of effectiveness like ``are these peaks `good enough'?''~\cite{LocalSearch_Book1,LocalSearch_Book3}.
Such questions of effectiveness can be avoided by focusing on just single-peaked fitness landscapes.
Or even more stringently, by focusing on fitness landscapes on the hypercube of assignments that are single peaked on every subcube.
These landscapes are known -- depending on the research community -- as semismooth fitness landscapes~\cite{evoPLS}, completely unimodal pseudo-Boolean functions~\cite{completelyUnimodal}, or acyclic unique-sink orientations (AUSOs) of the hypercube~\cite{AUSO_Thesis,AUSO}.
Semismooth fitness landscapes have a single peak, a short ascent from any starting point to this peak, and have nothing resembling other local peaks to potentially block local search (\cref{prop:semismooth,prop:semismooth_shortAscents}; Hammer et al.~\textcite{completelyUnimodal}, Poelwijk et al.~\textcite{rse2}, and Kaznatcheev~\textcite{evoPLS}; see \cref{sec:background}).
Local search is effective in finding the best possible result in semismooth landscapes, but can it do so efficiently?
 
The (in)efficiency of finding local peaks for general combinatorial optimization problems is captured by the complexity class of polynomial local search (\PLS;~\cite{PLS}):
if $\FP \neq \PLS$ then there is no polynomial time algorithm (even non-local) that can 
find a local peak in general fitness landscapes.
Many problems that generate these hard landscapes are complete under \emph{tight} \PLS-reductions, so have families of instances and initial assignments such that all ascents 
(i.e., sequences of adjacent assignments with strictly increasing fitness; \cref{sec:background}) 
to any local peak are exponentially long~\cite{PLS_Survey,PLS}. 
Thus, 
ascent-following local search will not find a local peak in polynomial time, even if $\FP = \PLS$.\footnote{
Restricting \PLS-complete problems to a subset of instances that are tractable in polynomial time does not mean that particular local search algorithms will also find a peak in polynomial time.
For example, \wSAT\ is a \PLS-complete problem~\cite{W2SAT_PLS1,W2SAT_PLS2} and thus binary Boolean valued constraint satisfaction problems (\VCSP) are \PLS-complete.
For \VCSP-instances of bounded treewidth, the global peak can be found in polynomial time~\cite{BB73,VCSPsurvey}.
Thus, bounded-treewdith \VCSP\ is not \PLS-complete.
But exponentially long ascents exist (and greedy local search algorithms like \lang{SteepestAscent} can take exponentially long) even with binary Boolean \VCSP-instances of pathwidth 2~\cite{repCP,tw7,pw4,pw2,pw2MSc}.
What is easy for non-local search is not necessarily easy for local search.
}
In contrast, semismooth fitness landscapes always have some short ascent to their unique fitness peak.
Thus, any results about (in)efficiency will have to be stated in terms of specific local search algorithms.\footnote{
The special case of finding the peak in semismooth landscapes is reducible to \lang{Unique\text{-}End\text{-}of}\-\lang{\text{-}Potential\text{-}Line} -- the complete problem for $\UEOPL \subseteq \PLS \cap \PPAD$~\cite{UEOPL,EOPLeqPLScapPPAD}.
It is believed that \UEOPL\ is strictly easier than \PLS, but still not tractable in polynomial time.
In a blackbox setting (no access to a concise description of the fitness landscape), \UEOPL\ and \PLS\ are not tractable in polynomial time.
}
And there exist constructions of semismooth landscapes that cannot be solved efficiently by various popular local search algorithms including \lang{SteepestAscent}~\cite{fittestHard3,fittestHard2,evoPLS,pw2}, \lang{RandomAscent}~\cite{randomFitter2,randomFitter1}, \lang{RandomFacet}~\cite{randomFacetBound},
jumping rules~\cite{LongJump}, and many others.
It is generally believed that for any particular local search algorithm, there will be some family of semismooth fitness landscapes that show that the algorithm is not efficient.
On semismooth fitness landscapes, even if local search is always effective, it is not always efficient.

This raises our main question for this paper: \emph{when is local search efficient on single-peaked landscapes?}
In a classic complexity theory setting, one wants to keep fixed a given problem (or class of problem instances) and seek for the simplest or most natural algorithm that solves this problem effectively (i.e., correctly) and efficiently (i.e., in polynomial time).
We flip this formula around.
We fix the algorithms and seek the most complex class of problem instances on which these algorithms can solve the problem effectively and efficiently. 
Specifically, we fix a collection of many popular local search algorithms and
seek to find a large class of single-peaked fitness landscapes on which this collection of many local search algorithms is efficient (we abbreviate this as \emph{efficient-for-many}; and we let our fixed collection be algorithms from a broader algorithm-class that we name poly-bypass).

Given that blackbox fitness landscapes are unwieldy exponentially large objects, we open the blackbox by representing fitness landscapes by instances of valued constraint satisfaction problems (\VCSP s)~\cite{repCP}.\footnote{
Our focus on easily checkable properties of polynomial-sized representations of problems instead of purely theoretical properties of the exponentially-large fitness landscapes implicit in the problems, is one of the big differences of our approach/results versus similar work in evolutionary computation~\cite{EvoComp2,LocalSearch_Book4,EvoComp3}.
}
Specifically, we study binary Boolean \VCSP s, also known as quadratic pseudo-Boolean functions.
Since arbitrary binary Boolean \VCSP s are \PLS-complete~\cite{PLS_Survey,PLS,W2SAT_PLS1,W2SAT_PLS2},
we will be looking for subclasses of binary Boolean \VCSP{s} that implement single-peaked fitness landscapes on which poly-bypass local search algorithms can find the peak efficiently.
Thus, our goal is to find an easy to check structural property of \VCSP{s} that captures the most expressive subclass of binary Boolean \VCSP{s} for which many popular local search algorithms are both effective and efficient.

\section{Summary of Results}

To avoid concerns over (in)effectiveness, we start \cref{sec:AUSOfromVCSP} by identifying the binary Boolean \VCSP s that implement semismooth fitness landscapes.
We do this by assigning each edge in the constraint graph of a binary Boolean \VCSP\ corresponding to a valued constraint between variable $x_i$ and $x_j$ one of two types: directed ($\se{i}{j}$ or $\se{j}{i}$) or bidirected ($\rse{i}{j}$).
We show that \VCSP s without bidirectional edges are equivalent to semismooth fitness landscapes 
and name them \emph{directed} \VCSP s.
Given many constructions of hard semismooth fitness landscapes \cite{fittestHard3,randomFacetBound,randomFitter2,fittestHard2,evoPLS,randomFitter1,LongJump}, we do not expect local search to be efficient on directed \VCSP s.

To find efficiency, we define \emph{oriented} \VCSP{s} as a further restriction on directed \VCSP{s}.
For each pair of variables $x_i$, $x_j$ in an oriented \VCSP\ there can be at most one of $\se{i}{j}$ or $\se{j}{i}$:
the constraint graph is oriented, hence the name.
We show that there are no directed cycles in oriented \VCSP s: all oriented \VCSP s induce a partial order where the preferred assignment of a variable $x_j$ depends only on the assignments of variables $x_i$ with $i$ lower than $j$ in the partial order.
This means that once we condition on all $x_i$ with $i$ lower than $j$, the preferred assignment of $x_j$ is independent of all other variables -- like in a smooth landscape.

Although having a directed or oriented constraint graph is a natural property to define subclasses of \VCSP s, it is a hard property to recognize.
Specifically, 
\lang{SubsetSum} can be solved by determing the direction of constraints (\cref{prop:rse_NPC} for $\rse{i}{j}$ and \cref{prop:doublese_NPC} for both $\se{i}{j}$ and $\se{j}{i}$).
Thus, given an \textit{arbitrary} binary Boolean \VCSP, the problem of checking if it is directed or oriented is \coNP-complete (\cref{cor:semismoothTesting_coNPC}).
Furthermore, even if we are given a \textit{directed} binary Boolean \VCSP, checking if it is oriented is \coNP-complete (\cref{cor:orientedTesting_coNPC}).
The only silver lining is that if we take the maximum degree of the constraint graph of the \VCSP-instance as a parameter then checking if a binary Boolean \VCSP\ is directed or oriented is fixed-parameter tractable (\cref{alg:VCSP_arcs} and \cref{sec:complexity_arcs}).

Given the importance of the partial order of conditional independence but the difficulty of recognizing if a \VCSP\ is oriented, with \cref{def:precSmooth} in \cref{sec:precSmooth}, we abstract to a class of landscapes that we call conditionally-smooth.
Just like smooth and semismooth landscapes, conditionally-smooth landscapes have only one local (and thus global) peak and there exist short ascents from any initial assignment to the peak (\cref{prop:precSmooth_SinglePeaked}).
In other words, just as with semismooth landscapes, local search algorithms are effective on conditionally-smooth landscapes.
Unlike directed or oriented \VCSP s, however, we show how to recognize if an arbitrary binary Boolean \VCSP\ is conditionally-smooth in polynomial time (\cref{alg:precSmoothCheck}).

In \cref{sec:efficientAlgos}, we show that --
unlike general semismooth fitness landscapes (and like smooth landscapes and oriented \VCSP s)
-- conditionally-smooth fitness landscapes are not just effective-for but also efficient-for-many local search algorithms.\footnote{
We do not focus on \emph{efficient-for-all} (or more formally: landscapes where \emph{all} ascents are polynomial length) because we think that this case is too strict and has been largely resolved. 
Kaznatcheev, Cohen and Jeavons~\textcite{repCP} showed that all ascents have length $\leq \binom{n+1}{2}$ in fitness landscapes that are implementable by binary Boolean \VCSP\ with tree-structured constraint graphs.
Efficient-for-all cannot be pushed much further than Boolean trees: 
Kaznatcheev, Cohen and Jeavons~\textcite{repCP} also gave examples of \VCSP{s} with exponential ascents from
(a) domains of size $\geq 3$ and path-structured constraint graphs, or 
(b) Boolean domains and constraint graphs of pathwidth $2$. 
This is why, we switch from the question of efficient-for-all to the question of \emph{efficient-for-many}.
}
Specifically, we show that conditionally-smooth fitness landscapes are solved efficiently by many local search algorithms including \lang{RandomAscent} (\cref{prop:RA_stepsBound}); \lang{ShakenAscent} (\cref{prop:runningtime-shaken-ascent}); \lang{SimulatedAnnealing} (\cref{prop:SA_stepsBound}); \lang{ZadehsRule}, \lang{LeastRecentlyConsidered}, and other history-based local search (\cref{prop:HB_stepsBound}); \lang{AntipodalJump} and \lang{JumpToBest} (\cref{prop:runningtime-antipodal-and-best-jump}); \lang{KerninghanLin} (\cref{prop:KL_bound}); and \lang{RandomJump} (\cref{prop:randomJump_steps}) --
all of these are examples a broader class of local search algorithms that we name \emph{poly-bypass} local search algorithms (\cref{def:poly-bypass}).
Conditionally-smooth landscapes are a more expressive class than the tree-structured binary Boolean \VCSP{s} that (fully) capture the class of fitness landscapes that are efficient-for-all local search algorithms~\cite{repCP}.
Thus, the conditionally-smooth structural property permits families of landscapes that are complex enough to break the efficiency of some popular local search algorithms that are not \emph{poly-bypass} algorithms.
In Section \ref{sec:inefficientAlgos}, we show that there are families of conditionally-smooth landscapes where finding the peak takes an exponential number of steps for \lang{SteepestAscent} (\cref{cor:HakenLubyBound}) and takes a superpolynomial number of steps for \lang{RandomFacet}  (\cref{cor:RandomFacetSlow}).
Given that both \lang{SteepestAscent} and \lang{RandomFacet}~\cite{randomFacetBound} are often considered to be very good local search algorithms,\footnote{
\lang{RandomFacet} is currently considered the best known algorithm for solving semismooth fitness landscapes.
It finds the peak in any semismooth fitness landscapes in a superpolynomial but subexponential number of steps -- even with the landscape given as a black-box~\cite{randomFacetBound}.
In \cref{cor:MatousekPrecSmooth}, we show that the family of semismooth landscapes that saturate this worst case behavior for \lang{RandomFacet} are conditionally-smooth.
} 
their inefficiency on conditionally-smooth landscapes tells us that conditionally-smooth landscapes are not a trivially easy to solve class of landscapes.

\section{Effectiveness from smooth to semismooth fitness landscapes}
\label{sec:background}

We consider \emph{assignments} $x \in \{0,1\}^n$ from the \emph{$n$-dimensional hypercube} 
where $x_i$ refers to the $i$-th entry of $x$.
To refer to a substring with indexes $S \subseteq [n]$, we will write the partial assignment $y = x[S] \in \{0,1\}^S$.
If we want to modify a substring with indexes $S \subseteq [n]$ to match a partial assignment $y \in \{0,1\}^S$ (or any string $y \in \{0,1\}^{|S|}$), we write $x[S \mapsto y]$.
We abbreviate $x[\{i\} \mapsto b]$ by just $x[i \mapsto b]$.
Two assignments are adjacent if they differ on a single bit: $x,y \in \{0,1\}^n$ are adjacent if there exists an index $i \in [n]$ such that $y=x[i \mapsto \bar{x}_i]$,
where $\bar{x}_i := 1 - x_i$ is the negation of $x_i$.

This combinatorial structure of adjacent assignments can be combined with the numeric structure of a pseudo-Boolean function (that we call the \emph{fitness function}) to create a \emph{fitness landscape}~\cite{W32,evoPLS,repCP}.
A fitness landscapes $f$ associates to each assignment $x$ the integer $f(x)$ and a set of adjacent assignments.
Given an assignment $x$, we let $\phi^+(x) = \{i \; | \; i \in [n] \text{ and } f(x[i \mapsto \bar{x}_i]) > f(x)\}$ be the indexes of variables that increase fitness when flipped (out-map) and $\phi^-(x) = \{i \; | \; i \in [n] \text{ and } f(x) \geq f(x[i \mapsto \bar{x}_i]) \}$ be the set of indexes that lower or do not increase fitness when flipped (in-map).
An assignment $x$ is a \emph{local peak} in $f$ if for all $y$ adjacent to $x$ we have $f(x) \geq f(y)$ (i.e., if $\phi^+(x)$ is empty).\footnote{
\label{fn:peak}Note that local peaks are not necessarily `strict', they can have adjacent assignments of equal (but not greater) fitness.
So what some call a `fitness plateau' is for us a collection of adjacent local peaks.
This transforms how we think about popular `hard' landscapes like \lang{Needle} from the evolutionary computation literature.
A landscape $f$ is a \lang{Needle} landscape if there exists a single assignment $x^\text{needle}$ such that $f(x^\text{needle}) = 1$ and for all other assignments $x \neq x^\text{needle} \; f(x) = 0$.
We use the scare quotes around hard because for our definition of local peaks, \lang{Needle} landscapes are trivial: every assignment is either a local peak (for $x^\text{needle}$ and any $x \neq x^\text{needle}$ not directly adjacent to $x^\text{needle}$) 
or directly adjacent to a local peak (for $y$ directly adjacent to  $x^\text{needle}$)
and are thus solved by any local search algorithm in at most one step.
\lang{Needle} landscapes are only hard if we are looking for a global peak instead of any local peak. 
So from the perspective of our effective vs efficient distinction, many intractability results that rely on the difficulty of `navigating fitness plateaus or crossing fitness valleys' are statement of ineffectiveness rather than inefficiency.
This is one of the big differences between our approach and results, and the approach and results in the evolutionary computation literature~\cite{fitnessLevel1,EvoComp1,EvoComp2,LocalSearch_Book4,EvoComp3,fitnessLevel2}.
}
A sequence of assignments $x^0,x^1,\ldots,x^T$ is an \emph{ascent} if every $x^{t - 1}$ and $x^t$ are adjacent with the latter having higher fitness (i.e., $\forall t \leq T \; x^{t} = x^{t - 1}[i \mapsto \overline{x^{t - 1}_i}] \text{ and } i \in \phi^+(x^{t - 1})$) and $x^T$ is a local peak.

A fitness landscape $f$ is \emph{smooth} if for each $i \in [n]$ there exists an assignment $x^*_i$ such that for all assignments $x$ we have $f(x[i \mapsto x^*_i]) > f(x[i \mapsto \overline{x^*_i}])$.
In other words, in a smooth landscape, each variable $x_i$ has a preferred assignment $x^*_i$ that is independent of how other variables are assigned.
It is easy to see that all ascents are short (i.e, $\leq n$) in smooth fitness landscapes, so smooth landscapes are effective- and efficient-for-all local search algorithms.

We can relax the definition of smooth while maintaining effectiveness: a fitness landscape on the hypercube is a \emph{semismooth fitness landscape} if every subcube is single peaked~\cite{evoPLS}.
These are also known as completely unimodal pseudo-Boolean functions~\cite{completelyUnimodal}, or acyclic unique-sink orientations (AUSOs) of the hypercube~\cite{AUSO_Thesis,AUSO}.
Semismooth fitness landscapes have a nice characterization in terms of the biological concept of sign epistasis~\cite{evoPLS,rse2,rse1,se}:

\begin{definition}[Kaznatcheev, Cohen and Jeavons \textcite{repCP}]\label{def:sign-depends}
    We say index $i$ \emph{sign-depends} on $j$ in background $x$ (and write  $\se[x]{j}{i}$) if 
    $f(x[i\mapsto \Bar{x}_i]) > f(x)$ and 
    $f(x[\{i,j\} \mapsto \Bar{x}_i\Bar{x}_j]) \leq f(x[j\mapsto \Bar{x}_j])$
     If there is no background assignment $x$ such that $\se[x]{j}{i}$ then we say that $i$ \emph{does not sign-depend} on $j$ (write $\nse{j}{i}$).
    If for all $j\neq i$ we have $\nse{j}{i}$ then we say that $i$ is \emph{sign-independent}.
\end{definition}

This terminology of `sign' comes from the observation that the sign of the fitness effect of a change in $x_i$ depends on the value of $x_j$.
Since the sign of the fitness effect of a change in $x_i$ just indicates the preferred assignment of $x_i$ (with a positive sign indicating $x^*_i = 1$ and negative indicating $x^*_i = 0$), it is easy to link this definition to smooth landscapes:
a smooth landscape is a fitness landscapes where all indexes are sign-independent.
Or stated in the negative, a smooth landscape is one without sign-dependence.
We will relax this negative definition to get semismooth landscapes by instead excluding the concept of `reciprocal sign epistasis'~\cite{rse1,rse2}:

\begin{definition}[Poelwijk et al.\textcite{rse2}]\label{def:rse}
    If there exists a background assignment $x$ such that $\se[x]{j}{i}$ and $\se[x]{i}{j}$ then we say that $i$ and $j$ have \emph{reciprocal sign epistasis} in background $x$ (and write $\rse[x]{i}{j}$).
    If there is no background assignment $x$ such that $\rse[x]{i}{j}$ then we say $i$ and $j$ \emph{do not have reciprocal sign epistasis}, and use the symbol $\nrse{i}{j}$.
\end{definition}

\noindent If the background $x$ is clear form context or not important then we drop the superscript in the above notations and just write $\se{j}{i}$ or $\rse{i}{j}$.
The absence of reciprocal sign epistasis is clearly necessary for a fitness landscape to be semismooth, but it is also sufficient:

\begin{proposition}[Hammer et al.\textcite{completelyUnimodal}, Poelwijk et al.\textcite{rse2}, and Kaznatcheev \textcite{evoPLS}]
A fitness landscape $f$ on $n$ bits is semismooth if and only if for all $i,j \in [n]$ $\nrse{i}{j}$.
\label{prop:semismooth}
\end{proposition}

Conveniently, semismooth landscapes always have a short ascent to the unique peak:

\begin{proposition}[Hammer et al.\textcite{completelyUnimodal}, and Kaznatcheev \textcite{evoPLS}]
A semismooth fitness landscape has only one local (thus global) peak at $x^*$
and given any initial assignment $x^0$, there exists an ascent to $x^*$ of Hamming distance (i.e., length $\leq n$).
\label{prop:semismooth_shortAscents}
\end{proposition}

\noindent Unlike smooth landscapes, however, not all ascents in semismooth fitness landscapes are short,
and finding and following a short ascent is not easy.
The Klee-Minty cube~\cite{KM72} 
is a construction of a semismooth fitness landscape on $\{0,1\}^n$ with an ascent of length $2^n - 1$.
As for local search algorithms, constructions exists such that semismooth landscapes are not tractable in polynomial time by various popular ascent-following local search algorithms including \lang{SteepestAscent}~\cite{fittestHard3,fittestHard2,evoPLS}, \lang{RandomAscent}~\cite{randomFitter2,randomFitter1}, \lang{RandomFacet}~\cite{randomFacetBound}, jumping rules~\cite{LongJump}, and many others.
Thus, local search is not efficient on \emph{all} semismooth fitness landscapes.
So we start our search for single-peaked landscapes that are efficient-for-many local search algorithms by looking for natural subclasses of semismooth fitness landscapes.

\section{Representing semismooth fitness landscapes by directed \VCSP s}
\label{sec:AUSOfromVCSP}
\label{sec:recognize_doVCSP}

As blackboxes, fitness landscapes are unwieldy exponentially large objects, so we open the blackbox by representing landscapes by instances of valued constraint satisfaction problems (\VCSP s)~\cite{repCP}.
We can then find which natural subclass of \VCSP{s} represents semismooth fitness landscapes along with an algorithms for checking when a given \VCSP-instance has these properties.
A Boolean \VCSP-instance is a set of constraint weights $\mathcal{C} = \{c_S\}$ 
where each weight $c_S \in \mathbb{Z} \setminus \{0\}$ has a \emph{scope} $S \subseteq [n]$.
This set of constraints \emph{implements} a pseudo-Boolean function: 
\begin{equation}
f(x) = \sum_{c_S \in \mathcal{C}} c_S \prod_{j \in S} x_j.
\end{equation}
If $|S| \leq 2$ for all constraints then we say the \VCSP\ is binary.
We also view every binary \VCSP-instance $\mathcal{C}$ as a \emph{constraint graph} with edges $\{i,j\} \in E(\mathcal{C})$ if $c_{ij} := c_{\{i,j\}} \in \mathcal{C}$
and a neighbourhood function 
$N_\mathcal{C}(i) = \{j \; | \; \{i,j\} \in E(\mathcal{C})\}$.
This constraint graph is a way of representing where the potential sign-dependencies are in the fitness landscape.

Given a binary Boolean \VCSP{} $\mathcal{C}$ on the whole $n$-dimensional hypercube, it is sometimes useful to consider the binary Boolean \VCSP{} $\mathcal{C}'$ restricted to just a subset of the indexes $R \subseteq [n]$ with the other variables in $S := [n] \setminus R$ fixed to some assignment $y \in \{0,1\}^{S}$ (i.e., restricted to the face $\{0,1\}^R y$).
This restricted \VCSP{} $\mathcal{C}'$ will have the same binary constraint as $\mathcal{C}$ (i.e., if $i,j \in R$ and $c_{ij} \in \mathcal{C}$ then $c_{ij} \in \mathcal{C'}$) but the unary constraints will change to what we call the \emph{effective unaries} that `absorb' the binary constraints that cross the $R$-$S$ cut:

\begin{definition}\label{def:effective-unary} 
    Given a variable index $i$ with neighbourhood $N(i)$ and $R\subseteq N(i)$ a set of indices we do not want to fix, we define the \emph{effective unary} $\hat{c}_i(x,R) = c_i+\sum_{j\in N(i)\setminus R}x_j c_{ij}$. 
    For simplicity, we write $\hat{c}_i(x)$ for $\hat{c}_i(x,\emptyset)$.
\end{definition}

\noindent If $x$ and $y$ are two assignments with $x[N(i)\setminus R] = y[N(i)\setminus R]$
then $\hat{c}_i(x,R)=\hat{c}_i(y,R)$.
We use this to overload $\hat{c}_i$ to partial assignments: given $T \supseteq N(i) \setminus R$ and 
$y \in \{0,1\}^T$, we interpret $\hat{c}_i(y,R)$ as $\hat{c}_i(y0^{[n] \setminus T},R)$. 
From this, it is easy to see that restricting our \VCSP{} $\mathcal{C}$ to $\mathcal{C}$' on the face $\{0,1\}^R y$ will change the unaries to $c'_i = \hat{c}_i(y,R)$ (for $i \in R$).

In terms of representation of smooth landscapes, it is clear than if all constraints in a \VCSP{} are unary then the resulting fitness landscape is smooth.
But even given a binary Boolean \VCSP-instance with non-unary constraints, it is easy to check if that \VCSP-instance implements a smooth landscape by checking if the unary constraints are ``sufficiently large'' compared to the relevant binary constraints.
Or more generally, given any partial assignment $y \in \{0,1\}^S$ it is easy to check if the face $\{0,1\}^{[n] - S}y$ is smooth using the effective unaries.
We will do this check one variable at a time.
For a variable $x_i$ to have a preferred assignment that is independent of how other variables are assigned, its unary must dominate over the binary constraints that it participates in -- its unary must have big magnitude.
For example, suppose that $c_i > 0$ then $x_i$ will prefer to be $1$ if all of its neighbours are $0$.
This preference must not change as any $x_j$ with $j \in N(i)$ change to $1$s.
Since $x_j$ with $c_{ij} < 0$ are the only ones that can lower $x_i$'s preference for $1$, this becomes equivalent to checking if
$c_i > \sum_{j \in N(i) \text{ s.t. } c_{ij} < 0}|c_{ij}|$.
There exists a combination of $x_j$s that flip $x_i$'s preference if and only if this inequality is violated.
In moving from this particular example to the general case, we can also generalize our algorithm to check not only if $x_i$ has an independent preference in the whole landscape but also if it has a conditionally independent preference conditional on some variables with indexes $j \in S$ not being able to vary and fixed to some $y \in \{0,1\}^S$.
To do this we only need to replace $c_i$ by the effective unary of \cref{def:effective-unary} $\hat{c}_i = \hat{c}_i(y,N(i) \setminus S) := c_i + \sum_{j \in S \cap N(i)}y_j c_{ij}$ and consider $c_{ij}$s with $\text{sgn}(\hat{c}_i) \neq \text{sgn}(c_{ij})$ and $j \in N(i) \setminus S$.
We encode this in  \lang{ConditionallySignIndependent}$(\mathcal{C},i,S,y)$ of \cref{alg:sign-independence}.

\begin{restatable}[!htb]{algorithm}{printAlgSI}
\caption{\lang{ConditionallySignIndependent}$(\mathcal{C},i,S,y)$:}
\label{alg:sign-independence}
\begin{algorithmic}[1]
    \Require \VCSP-instance $\mathcal{C}$, variable index $i$, subset of indexes $S\subset [n]$, an assignment $y$.
    
    \State Compute $\hat{c} \gets \hat{c}_i(y,N(i) \setminus S) = c_i + \sum_{j \in S \cap N(i)} y_jc_{ij}$
    \State Compute $c_\text{flip} \gets \sum_{j \in N(i) \setminus S} c_{ij}$ when $\text{sgn}(c_{ij}) \neq \text{sgn}(\hat{c})$.
    \State \textbf{return} $|\hat{c}|>|c_\text{flip}|$
\end{algorithmic} 
\end{restatable}

\noindent Given a binary Boolean \VCSP-instance $\mathcal{C}$, this algorithms runs in time linear in the maximum degree $\Delta(\mathcal{C})$ of the constraint graph.
To check if $\mathcal{C}$ implements a smooth landscape, run \lang{ConditionallySignIndependent}$(\mathcal{C},i,\emptyset,0^n)$ on every $i$ and return the conjunction of their outputs in $O(\Delta(\mathcal{C})n)$ time.

Now let us return to view the constraint graph as a way of represetning where the potential sign-dependencies are in a fitness landscape.
If \VCSP-instance with binary constraints is smooth then this tells us that the edges in the constraint-graph did not actually encode any actual sign-dependence.
This idea can be taken further by converting edges to arcs and assigning `directions' to the binary constraint $c_{ij}$ based on how its weight compares to the effective unaries $\hat{c}_i$ across various background assignments (see \cref{fig:arcDirections} for illustration): 

\begin{definition}
    For \VCSP-instance $\mathcal{C}$ we set the \emph{arcs} $A(\mathcal{C})$ such that for $\{i,j\} \in E(\mathcal{C})$:
    \begin{enumerate}[(1)]
        \item we set $\arc[b]{i}{j}\in A(\mathcal{C})$ if there exists an assignment $x$ with $|c_{ij}| > \max\{ |\hat{c}_i(x,\{i,j\})|, |\hat{c}_j(x,\{i,j\})| \}$ with $\text{sgn}(c_{ij}) \neq \text{sgn}(\hat{c}_i(x,\{i,j\}))$ and $\text{sgn}(c_{ij}) \neq \text{sgn}(\hat{c}_j(x,\{i,j\}))$.
        \item otherwise we set:
        \textbf{(a)} $\arc[r]{i}{j}\in A(\mathcal{C})$ if there exists an assignment $y$ such that $|c_{ij}| > |\hat{c}_j(y,\{i,j\})|$ with $\text{sgn}(c_{ij}) \neq \text{sgn}(\hat{c}_j(y,\{i,j\}))$, 
        and 
        \textbf{(b)} $\arc[l]{i}{j}\in A(\mathcal{C})$ if there exists an assignment $z$ such that $|c_{ij}| > |\hat{c}_i(z,\{i,j\})|$ with $\text{sgn}(c_{ij}) \neq \text{sgn}(\hat{c}_i(z,\{i,j\}))$.
    \end{enumerate}
    \label{def:VCSP_arcs}
\end{definition}

\noindent From this, each edge  $\{i,j\} \in E(\mathcal{C})$ is assigned only one of the three kinds of arcs: $\{\arc[b]{i}{j}\}$ or $\{\arc[r]{i}{j}$ and/or $\arc[l]{i}{j}\}$ or $\{\}$.
In \cref{fig:arcDirections}, we provide some minimal prototypical examples of \VCSP-instances that have (a) no direction for the constraint with scope $\{i,j\}$, (b) $\arc[r]{i}{j} \in A(\mathcal{C})$, (c) both $\arc[r]{i}{j}, \arc[l]{i}{j} \in A(\mathcal{C})$, or (d) $\arc[b]{i}{j} \in A(\mathcal{C})$.
Most of these examples are on two variables.
Since it is impossible to have both both $\arc[r]{i}{j}$ and $\arc[l]{i}{j}$ without a $k \in N(i) \cup N(j)$ (\cref{prop:unclustered_bi}), the minimal instance in \cref{fig:arcDirections2se} requires three variables.
It is easy to check that the overloading of the sign-dependence and reciprocal sign epistasis symbols is appropriate.
Most importantly, the fitness landscape will have reciprocal sign epistasis if and only if the \VCSP-instance $\mathcal{C}$ implementing it has $\{\arc[b]{i}{j}\} \in A(\mathcal{C})$.

\begin{figure}
    \hfill
    \begin{subfigure}[b]{0.22\textwidth}
        \begin{tikzpicture}[scale=0.75]
        \tikzstyle{every node}=[font=\small]
        \node[draw, circle, minimum size=4pt, label=left:$3$] (i) at (0,0) {$i$};
        \node[draw, circle, minimum size=4pt, label=right:$3$] (j) at (2,0) {$j$};
        \draw[thick,dashed] (i) -- (j) node [midway, above, sloped, fill=white] {$-2$};
        \path [use as bounding box] (0,-0.5) rectangle (1.5,2.0);
    \end{tikzpicture}
    \subcaption{$\{\} \in A(\mathcal{C})$}
    \end{subfigure}
    \hfill
    \begin{subfigure}[b]{0.22\textwidth}
    \begin{tikzpicture}[scale=0.75]
        \tikzstyle{every node}=[font=\small]

        \node[draw, circle, minimum size=4pt, label=left:$3$] (i) at (0,0) {$i$};
        \node[draw, circle, minimum size=4pt, label=right:$1$] (j) at (2,0) {$j$};
        \draw[->, thick] (i) -- (j) node [midway, above, sloped, fill=white] {$-2$};
        \path [use as bounding box] (-0.5,-0.5) rectangle (1.5,2.0);
    \end{tikzpicture}
    \subcaption{$\arc[r]{i}{j}\in A(\mathcal{C})$}
    \end{subfigure}
    \hfill
    \begin{subfigure}[b]{0.24\textwidth}
        \begin{tikzpicture}[scale=0.75]
        \tikzstyle{every node}=[font=\small]
            
            \node[draw, circle, minimum size=4pt, label=left:$1$] (i) at (0,0) {$i$};
            \node[draw, circle, minimum size=4pt, label=right:$3$] (j) at (2,0) {$j$};
            \node[draw, circle, thick, dotted, minimum size=4pt, label=above:$3$] (k) at (1,1.5) {$k$};
            \node[circle] (inv) at (1,0) {$-2$};
            
            \draw[->, thick] (i) to [out=30,in=150] (j);
            \draw[->, thick] (j) to [out=210,in=-30] (i);
            
            \draw[->, thick, dotted] (k) -- (i) node [midway, above, sloped, fill=white] {$2$};
            \draw[->, thick, dotted] (k) -- (j) node [midway, above, sloped, fill=white] {$-2$};
        \path [use as bounding box] (-0.5,-0.5) rectangle (1.5,2.0);
        \end{tikzpicture}
    \subcaption{$\arc[r]{i}{j},\arc[l]{i}{j}\in A(\mathcal{C})$}
    \label{fig:arcDirections2se}
    \end{subfigure}
    \hfill
    \begin{subfigure}[b]{0.24\textwidth}
        \begin{tikzpicture}[scale=0.75]
        \tikzstyle{every node}=[font=\small]
            
            \node[draw, circle, minimum size=4pt, label=left:$1$] (i) at (0,0) {$i$};
            \node[draw, circle, minimum size=4pt, label=right:$1$] (j) at (2,0) {$j$};
            \draw[<->, thick] (i) -- (j) node [midway, above, sloped, fill=white] {$-2$};
        \path [use as bounding box] (-0.5,-0.5) rectangle (1.5,2.0);
        \end{tikzpicture}
    \subcaption{$\arc[b]{i}{j}\in A(\mathcal{C})$}
    \end{subfigure}
    \caption{Four \VCSP{} instances illustrating the different arc directions of \cref{def:VCSP_arcs}.
    Weights of unary constraints are next to nodes and weights of binary constraints are above the edges.}
    \label{fig:arcDirections}
\end{figure}

The advantage of \cref{def:VCSP_arcs} over black-box features of the fitness landscape is its statement in terms of properties of just the \VCSP-instance.
This means, for example, that if we want to check the potential direction of a constraint with scope $\{i,j\} \in E(\mathcal{C})$ with $c_{ij} < 0$ then we need to compare it to the effective unaries $\hat{c}_i(x,\{i,j\})$ and $\hat{c}_j(x,\{i,j\})$ for various choices of $x$.
If we find an $x$ such that both $|c_{ij}| > \hat{c}_i(x,\{i,j\}) > 0$ and $|c_{ij}| > \hat{c}_j(x,\{i,j\}) > 0$ are simultaneously satisfied then we output that $\arc[b]{i}{j} \in A(\mathcal{C})$.
If we only ever satisfy one or fewer of these equations for all choices of $x$ then we need to output a subset of $\{\arc[r]{i}{j},\arc[l]{i}{j}\}$ depending on which of the comparisons was true.
Outputing $\{\}$ if neither comparison was ever satisfied, $\arc[r]{i}{j}$ if only the first was satisfied, $\arc[l]{i}{j}$ if only the second, and $\{\arc[r]{i}{j},\arc[l]{i}{j}\}$ if each was true in a different background.
Finally, to be fixed-parameter tractable, it is important to use that $\hat{c}_i(x,R)=\hat{c}_i(y,R)$ when $x[N(i)\setminus R] = y[N(i)\setminus R]$ to limit our search to just partial assignments $x \in \{0,1\}^{N(i) \cup N(j) \setminus \{i,j\}}.$
We formalize this as \cref{alg:VCSP_arcs}:

\begin{restatable}[!htb]{algorithm}{printAlgArcs}
\caption{\lang{ArcDirection}($\mathcal{C},i,j$):}
\label{alg:VCSP_arcs}
\begin{algorithmic}[1]
\Require \VCSP-instance $\mathcal{C}$ and variable indexes $i$ and $j$.
\State Initialize $A \gets \{\}$
\For{$x \in \{0,1\}^{N(i) \cup N(j) \setminus \{i,j\}}$}
\State Set $B \gets \{\}$, $\hat{c}_i = \hat{c}_i(x,\{i,j\})$, $\hat{c}_j = \hat{c}_j(x,\{i,j\})$
\State \algorithmicif\ $|\hat{c}_i| < |c_{ij}|$ and $\text{sgn}(\hat{c}_i) \neq \text{sgn}(c_{ij})$ \algorithmicthen\ $B \gets B \cup \{\arc[r]{i}{j}\}$
\State \algorithmicif\ $|\hat{c}_j| < |c_{ij}|$ and $\text{sgn}(\hat{c}_j) \neq \text{sgn}(c_{ij})$ \algorithmicthen\ $B \gets B \cup \{\arc[l]{i}{j}\}$
\State \algorithmicif\ $|B| < 2$ \algorithmicthen\ $A \gets A \cup B$
\algorithmicelse\ \textbf{return} $\{\arc[b]{i}{j}\}$
\EndFor
\State \textbf{return} A
\end{algorithmic}
\end{restatable}

The resulting worst-case runtime is  $2^{O(\Delta(\mathcal{C}))}$,
or an overall runtime of $2^{O(\Delta(\mathcal{C}))}O(\Delta(\mathcal{C})n)$ to determine the direction of all arcs in the \VCSP-instance $\mathcal{C}$.
This is fixed-parameter tractable when parameterized by the maximum number of constraints incident on a variable ($\Delta(\mathcal{C})$).
A fully polynomial time algorithm for finding arc directions is unlikely given that the questions ``is $\rse{i}{j} \in A(\mathcal{C})$?'' (\cref{prop:rse_NPC}) and ``are both $\se{i}{j}$ and $\se{j}{i}$ in $A(\mathcal{C})$?'' (\cref{prop:doublese_NPC}) are \NP-complete by reduction from \lang{SubsetSum}.
However, we can still define two natural subclasses of \VCSP s by restricting the bidirected constraint graph: 

\begin{definition}
    We say that a VCSP instance $\mathcal{C}$ is \emph{directed} if $\mathcal{C}$ has no bidirected arcs.
    We say that a directed $\mathcal{C}$ is \emph{oriented} if it has at most one arc for every pair of variables $i\neq j$.
    \label{def:VCSP_dir_ori}
\end{definition}

\noindent In \cref{sec:propVCSPs}, we extend \cref{prop:semismooth} to representations to show that
a quadratic pseudo-Boolean $f$ is semismooth if and only if the corresponding \VCSP{} is directed (\cref{prop:directedVCSP_semismooth}), 
any triangle-free directed \VCSP\ is oriented (\cref{lem:triangle-free-means-oriented}),
and oriented \VCSP s are always acyclic (\cref{prop:no-cycles}).
\cref{lem:triangle-free-means-oriented,prop:no-cycles} let us view oriented \VCSP{s} as a kind of generalization of of the tree-structured \VCSP{s} that Kaznatcheev, Cohen, and Jeavons~\textcite{repCP} showed are efficient-for-all local search algorithms.
Specifically, we replace the undirected acyclicity of trees by the directed acyclicity of DAGs.
Unfortunately, checking if a \VCSP\ is directed or oriented is \coNP-complete (\cref{cor:semismoothTesting_coNPC} and \cref{cor:orientedTesting_coNPC}).

\section{From oriented \VCSP{s} to conditionally-smooth fitness landscapes}
\label{sec:precSmooth}

One of these nice features of our two natural subclasses of directed and oriented \VCSP{s} is that they let us capture when fitness landscapes are both effective-for-all and efficient-for-many local search algorithms.
Directed \VCSP{s} capture the semismooth fitness landscapes that are effective-for-all local search algorithms, but that have no known efficient algorithms.
Oriented \VCSP{s} are then a further restriction to capture those semismooth fitness landscapes that are efficient-for-many local search algorithms.
By focusing on the main aspect of acyclicity that makes oriented \VCSP{} tractable, we can generalize that class to a slightly larger class of conditionally-smooth fitness landscapes that are also single-peaked (and so effictive-for-all local search algorithms) but also recognizable from the implementing \VCSP{}-instance.

The acyclicity of an oriented \VCSP-instance (\cref{prop:no-cycles}) lets us define a strict partial order $\prec$ as the transitive closure of the constraint graph and the corresponding down sets $\downarrow j = \{i\; | \; i \prec j\}$.
As we show later, what makes oriented \VCSP{s} efficient for many local search algorithms is that this order respects conditional independence (from \cref{alg:sign-independence}).
Specifically \cref{prop:orientedVCSPcombed}: for oriented \VCSP-instance $\mathcal{C}$, $\forall y \in \{0,1\}^{\downarrow j}$ $\lang{Conditionally SignIndependent}(\mathcal{C},j,\downarrow j,y) = \textsc{True}$. 
\cref{prop:orientedVCSPcombed} is a powerful defining feature of oriented \VCSP{s}, but it is more powerful (and restrictive) than necessary for efficiency.
To get our definition of the larger class of \emph{conditionally-smooth fitness landscapes}, we can relax from conditional sign independence for any background $y \in \{0,1\}^{\downarrow j}$ to just a single background $y = x^*[\downarrow j]$ where $x^*$ is the peak of a single-peaked landscape:

\newpage
\begin{definition}
Given a strict partially ordered set $([n],\prec)$ and $\downarrow j = \{i\; | \; i \prec j\}$,
we call a fitness landscape $f$ on $\{0,1\}^n$ a \emph{$\prec$-smooth fitness landscape} with optimum $x^*$ 
when for all $j \in [n]$ and $x \in \{0,1\}^n$, 
if $x[\downarrow j] = x^*[\downarrow j]$ then $f(x[j \mapsto x^*_j]) > f(x[j \mapsto \overline{x^*_j}])$.\footnote{
Thus $\emptyset$-smooth fitness landscapes are just smooth fitness landscapes.
}
We say $f$ is a \emph{conditionally-smooth} fitness landscape if there exists some $\prec$ such that $f$ is $\prec$-smooth.
\label{def:precSmooth}
\end{definition}

\noindent Conditionally-smooth landscapes generalize both oriented \VCSP s and recursively combed AUSOs (\cref{def:rcAUSO}).
Although conditionally-smooth landscapes are not always semismooth, they are singled peaked and have direct ascents from any assignment to the peak:

\begin{proposition}
A conditionally-smooth fitness landscape has only one local (thus global) peak at $x^*$
and given any initial assignment $x^0$, there exists an ascent to $x^*$ of Hamming distance (i.e., length $\leq n$).
\label{prop:precSmooth_SinglePeaked}
\end{proposition}
\begin{proof}
    Let $f$ be a $\prec$-smooth landscape. Take any assignment $x\neq x^*$ and let $i$ be the $\prec$-smallest index such that $x_i\neq x_i^*$.
    Then $f(x[i\mapsto\overline{x_i}])>f(x)$.
\end{proof}

\noindent Unlike with directed or oriented \VCSP s, we can check if a binary Boolean \VCSP\ instance implements a fitness landscape that is conditionally smooth in polynomial time.
In fact, if the \VCSP\ does implement a conditionally-smooth fitness landscape then our recognition \cref{alg:precSmoothCheck} even returns a partial order $\prec$ and the preferred assignment $x^*$ 
such that $\mathcal{C}$ is $\prec$-smooth.
The overall approach is similar to checking if a \VCSP\ is smooth (i.e., $\emptyset$-smooth) with the difference being that subsequent calls to \lang{ConditionallySignIndependent}$(\mathcal{C},i,S,y)$ adjust the set of fixed variables $S$ and background assignment $y$ based on previous calls.
This gives the \lang{ConditionallySmooth}$(\mathcal{C})$ algorithm (\cref{alg:precSmoothCheck}).
This algorithm calls \lang{ConditionallySignIndependent} at most $\binom{n}{2}$ times for an overall runtime of $\binom{n}{2}\Delta(\mathcal{C})$ (for details of the analysis, see \cref{sec:precSmooth_rec}).

\begin{restatable}[!htb]{algorithm}{printAlgPrecSmooth}
\caption{\lang{ConditionallySmooth}($\mathcal{C}$)
\\ Checking if a \VCSP\ implements a $\prec$-smooth landscape.}
\label{alg:precSmoothCheck}
\begin{algorithmic}[1]
\Require \VCSP-instance $\mathcal{C}$
\State Initialize $\prec \; \gets \emptyset$, $S \gets \emptyset$, $x^* \gets 0^n$
\While{$S$ is not $[n]$}
\State $T \gets \emptyset$ and $x^\text{next} \gets x^*$
\For{$i \in [n] \setminus S$}
\If{\lang{ConditionallySignIndependent}$(\mathcal{C},i,S,x^*)$}
\State \algorithmicif\ $\hat{c}_i(x^*,N(i)\backslash S) > 0$ 
\algorithmicthen\ $x^\text{next}_i \gets 1$ \label{alg:precSmoothCheck_line:step}
\State $T \gets T + \{i\}$
\EndIf
\EndFor
\State \algorithmicif\ $T$ is empty 
\algorithmicthen\ \textbf{return} \textsc{False} \label{alg:precSmoothCheck_line:false}
\State $\prec \; \gets \; \prec + \; S \times T$, 
$S \gets S + T$, 
$x^*\gets x^\text{next}$
\EndWhile
\State \textbf{return} $(\prec,x^*)$
\end{algorithmic}
\end{restatable}

\section{Efficient local search in conditionally-smooth landscapes}
\label{sec:efficientAlgos}

Local search starts at some initial assignment $x^0$ and takes steps to assignments $x^1, x^2, \ldots , x^T$ with $x^T$ as a local peak.
If, additionally, for every $0 \leq t < T$ we have $f(x^{t + 1})>f(x^t)$ then local search followed an ascent.
For an arbitrary local search algorithm \lang{A} we let $\lang{A}^t_f(x)$ denote $t$ steps of \lang{A} from $x$ on fitness landscape $f$.

\begin{definition}
Given a polynomial $p(n)$, we say that an ascent-following\footnote{
One can modify \cref{def:poly-bypass} to apply to ascent-biased algorithms (\cref{app:effectiveFitnessDecreasing}) or jumping algorithms (\cref{app:efficientJumping}) instead of just ascent-following local search algorithms -- but this makes the definition more unwieldy and unintuitive.
We want to present poly-bypass algorithms as just a simple warm-up example to get at the main ideas of our later proofs.
So here we focus on just ascent-following poly-bypass for simplicity, and go into the nuance of stochastic, ascent-biased, and jumping rules when we focus on proving even tighter bounds for specific popular examples of those local search algorithms.
} local search algorithm $A$ is a \emph{$p(n)$-bypass ascent following local search algorithm} if given any $f$, any corresponding run $x^0, x^1, \cdots, x^T$ of $A$ on $f$, and all $s \in [T - p(n)]$, we have that $\cap_{t = s}^{s + p(n)} \phi^+(x^t)$ is empty (with high probability, for randomized algorithms).
If some polynomial $p(n)$ exists, but its specific form is not important to us, then we say that the algorithm is a \emph{poly-bypass} ascent-following local search algorithm.
\label{def:poly-bypass}
\end{definition}

\noindent For a $p(n)$-bypass algorithm, any index $i$ that could have flipped to increase fitness at some step $s$ (i.e., $i \in \phi(x^s)$), will become an index that cannot be flipped to increase fitness at some point in the subsequent $p(n)$ steps.
This can happen either because the variable with index $i$ flips or because some other indexes flip in a way that makes a flip at $i$ no longer fitness increasing.
In other words, no potential fitness-increasing flip was bypassed for more than $p(n)$ steps.
Hence, the name.

Now, we can easily show that conditionally-smooth landscapes are efficient for $p(n)$-bypass local search algorithms.

\begin{theorem}
On a $\prec$-smooth landscape $f$ on $n$ bits, given any initial assignment $x^0$ at Hamming distance $m (\leq n)$ to the fitness peak $x^*$, any $p(n)$-bypass local search algorithm starting at initial assignments $x^0$ takes at most $m \cdot p(n)$ steps to find the peak.
\label{thm:bypass_steps}
\end{theorem}

\begin{proof}
We prove this by induction on $m$.
For $m = 0$, $x^0 = x^*$ and local search finishes without taking a step.
Now we assume the inductive hypothesis is true for Hamming distance $\leq m - 1$ and show it is true if the Hamming distance between $x^0$ and $x^*$ is $m$. 
Let $i$ be the $\prec$-smallest index such that $x^0_i \not= x^*_i$.
Since the algorithm is $p(n)$-bypass, sometime by the $p(n)$th step, it will be at an assignment such that the $i$th bit doesn't want to flip.
Since $i$ was $\prec$-smallest index such that $x^0_i \not= x^*_i$ that means that the $i$th will always want to be in state $x^*_i$ and once it flips, it won't flip back because the algorithm is ascent-following, so $x^{p(n)}_i = x^*_i$.
Thus $x^{p(n)}$ has at least one more variable assignment in common with $x^*$ than $x^0$ did, so the Hamming distance between $x^{p(n)}$ and $x^*$ is $\leq m - 1$.
By the inductive hypothesis, the $p(n)$-bypass algorithm will find the fitness peak in at most $(m - 1)\cdot p(n)$ steps starting from $x^{p(n)}$.
This gives us a total number of steps of less than $m \cdot p(n)$.
\end{proof}

For a specific local search algorithm, the bound in \cref{thm:bypass_steps} can be rather loose.
To provide tighter bounds on the number of steps taken by many popular local search algorithms, we will need a bit more fine-grained notation and more careful proofs.
However, much like the proof of \cref{thm:bypass_steps}, all these proofs will rest on showing a bound on how long it takes to fix $\prec$-minimal indexes that disagree with $x^*$ and repeatedly applying that bound until the assignments agree (\cref{thm:main_stepsBound}).

To show that conditionally-smooth landscapes (and thus also oriented \VCSP{s}) are efficient for many popular local search algorithms with tighter bounds, we need to state \cref{thm:main_stepsBound} precisely.
This requires us to create a partition of the $n$ indexes and show that many local search algorithms quickly and ``permanently'' fix variables with indexes progressing along the levels of this partition.
Given a  poset $([n],\prec)$ and variable index $i \in [n]$, define the upper set of $i$ as $\uparrow i = \{j \;|\; i \preceq j \}$.
We partition $[n]$ into $\text{height}([n],\prec)$-many \emph{level sets} defined as $S_l = \{ i \;|\; \text{height}(\uparrow i, \prec) = l\}$ where $\text{height}(S,\prec)$ is the height of the poset $(S,\prec)$.
Additionally we define $S_0=\emptyset$, $S_{<l}=\bigcup_{k=0}^{l-1}S_k$, and $S_{>l} = \bigcup_{k=l + 1}^{n}S_k$.
We relate these partitions to an assignment $x$ via a conditionally-smooth landscape specific refinement of in- and out-maps:
\begin{restatable}{definition}{defInOutMap}\label{def:in-out-maps}
Given a $\prec$-smooth landscape $f$ on $n$ bits and $x$ an assignment,
we define the maps $\phi^\ominus,\phi^\oplus: \{0,1\}^n \rightarrow 2^{[n]}$ by:
\begin{itemize}
\item $i \in \phi^\ominus(x)$ if for all $j \preceq i$ we have $j \in \phi^-(x)$
(we say $i$ is \emph{correct} at $x$), and
\item $i \in \phi^\oplus(x)$ if $i \in \phi^+(x)$ but for all $j \prec i$ we have $j \in \phi^-(x)$ 
(we say $i$ is at \emph{border} at $x$).
\end{itemize}
\end{restatable}
\leavevmode\newline
Note that 
given any ascent $x^0,x^1, \ldots , x^T$ in a conditionally-smooth landscape, we have $\phi^\ominus(x^0) \subseteq \phi^\ominus(x^1) \subseteq \ldots \phi^\ominus(x^T) = [n]$.
We refer to the set $\overline{\phi^\ominus(x)} := [n]-\phi^\ominus(x)$ as the \emph{free indices} at $x$.
For an assignment $x$, define $\text{height}_f(x)$ (and $\text{width}_f(x)$) as the height (and width) of the poset $(\overline{\phi^\ominus(x)},\prec)$.
Note that for $x$ with $\text{height}(x)=l$,\footnote{
If we define $X_l$ as the set of all assignments $x$ with $\text{height}(x)=l$, it is important to note that the resulting sets $X_0,\ldots,X_{\text{height}(\prec,[n])}$ are not necessarily monotonic in fitness.
There can exist $x$ at $\text{height}(x)=l$ and $y,z$ with height $l - 1$ such that $f(y) < f(x) < f(z)$.
Thus, even if we expressed our results in terms of sets of assignments (i.e., in the space of the fitness landscape rather than the more compact space of the representation that we use) our approach in this paper is still \emph{not} a fitness-level method 
of the sort often used in the analysis of randomized search heuristics~\cite{fitnessLevel1,LocalSearch_Book4,fitnessLevel2}.
}
we have $S_{>l} \subseteq \phi^\ominus(x)$, $S_l \subseteq \phi^\ominus(x) \cup \phi^\oplus(x)$, and, for $x\neq x^*$, $\phi^\oplus(x) \cap S_l\neq\emptyset$.
In other words, if $x$ is at height $l \neq 0$ then all the variables with indexes at higher levels are set correctly, all free indexes at level $l$ are at the border, and at least one index at level $l$ is free.
This also means that $\overline{\phi^\ominus(x[\phi^\oplus(x) \mapsto \overline{x[\phi^\oplus(x)]}])} \subseteq S_{<l}$.

\begin{restatable}{lemma}{mainStepsBound}\label{thm:main_stepsBound}
Given a $\prec$-smooth fitness landscape $f$ on $n$ bits and any assignment $x$ with $l := \text{height}_f(x)$,
let $Y: \Omega \times \mathbb{N} \rightarrow \{0,1\}^n$ be a stochastic process such that $Y_t \sim \lang{A}^t_f(x)$ is the random variable of outcomes of $t$ applications of the local search step algorithm $\lang{A}^1_f$ starting from $x$ and let the random variable $\tau_{<l}(\omega) := \inf \{t \;|\; \overline{\phi^\ominus(Y_t(\omega))} \in S_{<l}\}$ be the number of steps to decrement height. 
If the expected number of steps to decrement height is:
\begin{equation}
\mathbb{E}\{\tau_{<l}(\omega)\} \leq p(n,l) \leq q(n)
\label{eq:drift_time_bound}
\end{equation}
then the expected total number of steps taken by $\lang{A}$ to find the peak from an initial assignment $x^0$ is 
$\leq \sum_{l = 1}^{\text{height}_f(x^0)} p(n,l) \leq \text{height}_f(x_0) q(n) \leq \text{height}(\prec) q(n)$.
\end{restatable}
\begin{proof}
    This follows by induction on $\height[x^0]$ and linearity of expected value.
\end{proof}

\noindent We use \cref{thm:main_stepsBound} to show that conditionally-smooth fitness landscapes are efficient-for-many local search algorithms, whether they 
follow ascents (\lang{RandomAscent}; \cref{prop:RA_stepsBound}), 
occasionally step to adjacent assignments of lower fitness (\lang{SimulatedAnnealing}; \cref{prop:SA_stepsBound}), 
or even if they step to non-adjacent assignments of higher fitness (\lang{RandomJump}; \cref{prop:randomJump_steps}).

Since \cref{thm:main_stepsBound} is expressed in terms of a stochastic process, we begin by applying it to the prototypical stochastic local search algorithm: \lang{RandomAscent}~\cite{randomFitter2,simplexSurvey,randomFitter1}.
Given an assignment $x$, the step $\lang{RandomAscent}^1_f$ simply returns a fitter adjacent assignment uniformly at random.
Formally, if $Y^\text{RA}(x) \sim \lang{RandomAscent}^1_f(x)$ then:
\begin{equation}
\text{Pr}\{Y^\text{RA}(x) = x[i \mapsto \bar{x_i}]\} = \begin{cases}
\frac{1}{|\phi^+(x)|} & \text{ if } i \in \phi^+(x) \\
0 & \text{otherwise}
\end{cases}
\end{equation}

\noindent By bounding the expected number of steps to decrement height (\cref{lem:p_RA}) and  applying \cref{thm:main_stepsBound} we bound the expected total number of steps for \lang{RandomAscent}:

\begin{restatable}[\cref{app:efficientWalking}]{proposition}{propRAstepsBound}\label{prop:RA_stepsBound}
On a $\prec$-smooth landscape $f$ on $n$ bits, the expected number of steps taken by \lang{RandomAscent} to find the peak from initial assignment $x^0$ is: 
\begin{align}
\leq |\overline{\phi^\ominus(x^0)}| + \text{width}_f(x^0)(1 + \log \text{width}_f(x^0)) \binom{\text{height}_f(x^0) - 1}{2} \label{eq:RA_x0bound}\\
\leq n + \text{width}(\prec)(1 + \log \text{width}(\prec))\binom{\text{height}(\prec) - 1}{2}.
\end{align}
\end{restatable}

\noindent Good bounds are also possible for deterministic ascent-following algorithms like various history-based pivot rules~\cite{historyRules,cunningham1979,leastRecentlyEntered,zadeh} that we discuss in \cref{app:efficientHistory}.
But even if an ascent-following algorithm is not efficient on conditionally-smooth landscapes, it can become efficient through combination with \lang{RandomAscent}. 
Any ascent-following local search algorithm \lang{A} can be made into $\epsilon$-\lang{shaken-A} like this:
with probability $1 - \epsilon$ take take a step according to \lang{A}, and with probability $\epsilon$ take a step according to \lang{RandomAscent}.
Since variables with indexes in $\phi^\ominus(x)$ will never be unflipped by ascents, this combined algorithm's expected runtime is less than an $1/\epsilon$-multiple of the bound in \cref{prop:RA_stepsBound} (\cref{prop:runningtime-shaken-ascent}).

\cref{thm:main_stepsBound} also applies to algorithms that occasionally take downhill steps like simulated annealing~\cite{LocalSearch_Book1}.
Formally, if $Y^\text{SA}_{t + 1}(x) \sim \lang{SimulatedAnnealing}^1_f(x^t)$ then
\begin{equation}
\text{Pr}\{Y^\text{SA}_{t  +1} = y\} = \begin{cases}
\frac{1}{n} & \text{ if } i \in \phi^+(x) \text{ and } y = x[i \mapsto \overline{x_i}] \\
\frac{1}{n}\cdot r_t(f(x^t) - f(x^t[i \mapsto \overline{x^t_i}])) & \text{ if } i \in \phi^-(x) \text{ and } y = x[i \mapsto \overline{x_i}] \\
1 - \frac{|\phi^+(x)|}{n} - Z\frac{|\phi^-(x)|}{n} & \text{ if } y = x
\end{cases}
\end{equation}

\noindent where the downstep probability $r_t(\Delta f) \rightarrow 0$ monotonically as $t \rightarrow \infty$ for any $\Delta f > 0$
and $Z = \frac{1}{|\phi^-(x)|}\sum_{i \in \phi^-(x)} r_t(f(x^t) - f(x^t[i \mapsto \overline{x^t_i}]))$ is the average downstep probability.
A popular choice of downstep probability is $r_t(\Delta f) = \exp(\frac{\Delta f}{K(t)})$ with $K(t)$ a sequence of temperatures strictly decreasing to $0$.
But any downstep probability can be used to define a burn-in time $\tau^\alpha = \inf\{t \;|\; r_t(1) \leq \frac{\alpha}{n} \}$ that is a property of the algorithm and independent of the particular problem-instance (i.e., independent of the fitness landscape).

\begin{restatable}[\cref{app:effectiveFitnessDecreasing}]{proposition}{propSAstepsBound}\label{prop:SA_stepsBound}
On a conditionally-smooth fitness landscape $f$ on $n$ bits, the expected number of steps taken by \lang{SimulatedAnnealing} to find the peak is $\leq \tau^\alpha + n^2\frac{(\exp(\alpha) - 1)}{\alpha}$
where $\tau^\alpha = \inf\{t \;|\; r_t(1) \leq \frac{\alpha}{n} \}$.
\end{restatable}

\noindent \lang{RandomAscent} and \lang{SimulatedAnnealing} move to adjacent assignments, changing at most one variable at a time, and so require at least a linear number of steps.
But \cref{thm:main_stepsBound} also applies to local search algorithms that jump to non-adjacent assignments.
Some such algorithms like \lang{AntipodalJump}, \lang{JumpToBest}, and \lang{RandomJump} are especially popular on semismooth fitness landscapes since flipping any combination of variables with indexes in $\phi^+(x)$ results in a fitness increasing step~\cite{HowardsRule,simplexSurvey,LongJump}.
But algorithms that take steps to non-adjacent assignments are also used in other contexts like the Kernighan-Lin heuristic for \lang{MaxCut}~\cite{KL_Original}.
Such algorithms can exploit short but wide partial orders.
For example  on $\prec$-smooth landscapes, the number of steps taken by deterministic algorithms like \lang{AntipodalJump}, \lang{JumpToBest} and Kernighan-Lin are independent of $\text{width}(\prec)$, taking at most $\text{height}(\prec)$-steps (\cref{prop:runningtime-antipodal-and-best-jump} and \ref{prop:KL_bound}). 
Picking uniformly at random which subset of improving indexes to include in a jump -- as done by the like \lang{RandomJump} algorithm (see \cref{eq:randomJump} in \cref{app:efficientJumping}) -- requires only a $\log{(\text{width}(\prec)})$-factor more steps than the deterministic jump rules:

\begin{restatable}[\cref{app:efficientJumping}]{proposition}{propRJstepsBound}\label{prop:randomJump_steps}
    Let $f$ be a conditionally-smooth landscape and $x^0$ some assignment. The expected number of steps that \lang{RandomJump} takes to find the peak is at most $(\log(\width[x^0]) + 2)\height[x^0]$. 
\end{restatable}

Thus, we see that many local search algorithms can efficiently find the peak in conditionally-smooth fitness landscapes in a quadratic or fewer number of steps.

\section{Inefficient local search in conditionally-smooth landscapes}
\label{sec:inefficientAlgos}
\label{sec:steepestAscent}

Finally, in this section, we want to argue that the effectiveness of many local search algorithms on conditionally-smooth landscapes is not a trivial observation.
We do this by showing that conditionality-smooth landscapes are not tractable for some algorithms that are often considered to be very good local search algorithms (but that happen to not be poly-bypass).
Specifically, we will show that oriented \VCSP{s} can represent families of landscapes where \lang{SteepestAscent} takes an exponential number of steps; and that conditionally-smooth landscapes include families of semismooth landscapes where \lang{RandomFacet} takes a superpolynomial number of steps.
In other words, conditionally-smooth landscapes are expressive enough to contain very complicated kinds of fitness landscapes.

While studying Hopfield networks, Haken and Luby \textcite{hakenSteepest} created a family of binary Boolean \VCSP{s} with exponentially long steepest ascents.
Here we show that the Haken-Luby \VCSP{} is oriented and has pathwidth $3$.
The Haken-Luby \VCSP{} has variables with indexes $\{(k,i)|k\in[n], i\in [7]\}$ where each set of seven variables $\{(k,1),(k,2),\ldots (k,7)\}$ forms a gadget. 
The gadgets form a chain by connecting $(k,7)$ to $(k-1,1)$.
The magnitude of the constraints on the $k$th module is roughly proportional to $M_k=\frac{5}{6}(6^k-6)$ and the exact constraints, with the exception of $c_{(n,1)}=(6M_n+24)K$, where $K=2n+1$ is constant for the instance, are given by \cref{fig:steepestGadget}.
For example, $c_{(k,5),(k,7)}=-(2M_k+4)K$ and $c_{(k,7),(k-1,1)}=M_kK$.
In \cref{fig:steepestGadget} we also included the direction of the constraints.

\begin{figure}[htb]
    \centering
    
    \begin{tikzpicture}[every node/.style={minimum size=30pt,font=\small}]
    \def\w{M_k} 
    \def\C{K} 
    \def\e{\epsilon_k}
    \tikzmath{\xunit = 2.5; \yunit =3.5;}
        \node[align=center,draw] (M) at (2.5*\xunit,-1*\yunit) {$M_k=\frac{5}{6}(6^k-6)$\\$\e=n+1-k$\\$\C=2n+1$};
        
        \node[draw, dashed, circle] (inv) at (-2*\xunit,0) {\tiny $k+1,7$};
        
        \node[draw, circle, label={[label distance=0cm,rotate=-30]-90:$-(6\w+24)\C$}] (ANDin) at (-\xunit,0) {$k,1$};
        \node[draw, circle, label={[label distance=0cm]90:$-(3\w+10)\C-\e$}] (AND1) at (-.2*\xunit,\yunit) {$k,2$};
        \node[draw, circle, label={[label distance=0cm]-90:$-(3\w+11)\C$}] (AND2) at (-.2*\xunit,-\yunit) {$k,3$};
    
        \node[draw, circle, label={[label distance=0.0]-180:$-(2\w+7)\C$}] (XORmid) at (\xunit,0) {$k,5$};
        \node[draw, circle, label={[label distance=0cm]90:$-(3\w+9)\C$}] (XOR1) at (\xunit,\yunit) {$k,4$};
        \node[draw, circle, label={[label distance=0cm]-90:$-(3\w+9)\C$}] (XOR2) at (\xunit,-\yunit) {$k,6$};
        \node[draw, circle, label={[label distance=0.1cm]-90:$-(\w+1)\C$}] (XORjoin) at (2.3*\xunit,0) {$k,7$};
        \node[draw, dashed, circle 
        ] (XORout) at (3*\xunit,0) {\tiny $k-1,1$};
        
        \draw[->, thick] (ANDin) -- (AND1) node [midway, above=0.1cm, sloped] {$(3\w+10)\C+2\e$};
        \draw [<-, thick](AND2) -- (ANDin) node [midway, above, sloped] {$(3\w+12)\C$};
    
        \draw[->, thick] (AND1) -- (XOR1) node [midway, above, sloped] {$(3\w+10)\C$};
        \draw[->, thick] (AND2) -- (XOR2) node [midway, above, sloped] {$(3\w+10)\C$};
    
         \draw[<-, thick] (XORmid) -- (XOR1) node [midway, above, sloped] {$(2\w+6)\C$};
        \draw[->, thick] (XOR2) -- (XORmid) node [midway, above, sloped] {$(2\w+6)\C$};
    
        \draw[->, thick] (XOR1) -- (XORjoin) node [midway, above, sloped] {$(\w+2)\C$};
        \draw[->, thick] (XOR2) -- (XORjoin) node [midway, above, sloped] {$(\w+2)\C$};
        \draw[->, thick] (XORmid) -- (XORjoin) node [midway, above, sloped] {$-(2\w+4)\C$};
    
        \draw[->, thick] (XORjoin) -- (XORout) node [midway, above, sloped] {$\w\C$};
    
        \draw[->,dashed] (inv) -- (ANDin) node [midway, above=0.2cm, sloped] {$(\overbrace{6M_{k}+25}^{M_{k+1}})\C$};

    \end{tikzpicture}
    \caption{Haken-Luby gadget with $M_k=\frac{5}{6}(6^k-6)$, $\epsilon_k=n+1-k$, and $K=2n+1$. 
    Constraints of the $k$th of $n$ gadgets are shown:
    weights of unary constraints are next to their variables and weights of binary constraints are above the edges that specify their scope.
    Arcs are oriented according to Definition~\ref{alg:VCSP_arcs}, showing that the instance is oriented.
    Dotted arcs and vertices illustrate the connection to the neighboring gadgets. 
    For the boundaries: the unary of $(n,1)$ is $(6M_n+24)K>0$, $M_1=0$ and there is no binary constraint $c_{(1,7),(0,1)}$.}
    \label{fig:steepestGadget}
\end{figure}

\begin{proposition}
    Haken-Luby \VCSP\ is an oriented \VCSP.
    \label{prop:HL_oriented}
\end{proposition}
\begin{proof}
    It suffices to check the constraints in \cref{fig:steepestGadget} against \cref{def:VCSP_arcs}.
    For example, check the constraint with scope $S = \{(k,5),(k,7)\}$
    and weight $c_{S}=-(2M_k+4)K < 0$.
    First, look at $\hat{c}_{(k,5)}(x,S)$:
    among all partial assignments in $\{0,1\}^{\{(k,4),(k,6)\}}$, only the assignment that sets $x_{(k,4)}=1$ and $x_{(k,6)}=1$ yields a non-negative 
    $\hat{c}_{(k,5)}(11,S)=2(2M_k+6)K-(2M_k+7)K=(2M_k+5)K > 0$;
    since $|c_S| = (2M_k+4)K \not> (2M_k+5)K = \hat{c}_{(k,5)}(11,S)$ it follows that $\nse{(k,7)}{(k,5)}$.
    Second, look at $\hat{c}_{(k,7)}(x,S)$:
    the assignment that sets $x_{(k,4)}=1$,$x_{(k,6)}=0$ and $x_{(k-1,1)}=0$ yields 
    $\hat{c}_{(k,7)}(100,S)=((1-1)M_k+(2-1))K= K > 0$;
    since $|c_S| = (2M_k+4)K > K = \hat{c}_{(k,7)}(100,S)$, it follows that that $\se{(k,5)}{(k,7)}$.
    The other constraints can be checked similarly.
\end{proof}

\begin{proposition}
Haken-Luby \VCSP\ has pathwidth $3$.\footnote{
For standard definitions of pathwidth/treewidth, see \cite{parameterizedBook}.
}
\label{prop:HL_pw3}
\end{proposition}

\begin{proof}
($\Rightarrow$) Path decomposition for $k$-th gadget:
$\{(k+1,7),(k,1)\}$, $\{(k,1),(k,2),(k,3),(k,4)\}$, $\{(k,3),(k,4),(k,5),(k,6)\}$, $\{(k,4),(k,5),(k,6),(k,7)\}$, $\{(k,7),(k-1,1)\}$.

($\Leftarrow$) Contracting $(k,3)$, $(k,1)$, $(k,2)$ and $(k,4)$ yields a $K_4$ minor.
\end{proof}

\cref{prop:HL_oriented} and \ref{prop:HL_pw3} together with Haken and Luby's \textcite{hakenSteepest} proof that Haken-Luby \VCSP{s} have an exponential steepest ascent, gives us:

\begin{theorem}[Haken and Luby \textcite{hakenSteepest}]\label{cor:HakenLubyBound}
    There are oriented \VCSP{s} on $7n$ bits with constraint graphs of pathwidth $3$ such that \lang{SteepestAscent} follows an ascent of length $\geq 2^n$.
    \label{thm:steepestFail}
\end{theorem}

Unaware of much older Haken and Luby~\textcite{hakenSteepest}, Cohen et al.~\textcite{tw7} claimed to show pathwidth $7$ as best lower bound for a \VCSP\ with exponential steepest ascents, which was later improved to pathwidth $4$~\cite{pw4}.
Since \cref{prop:HL_pw3} shows that the Haken-Luby \VCSP\ has pathwidth $3$, this means it was the lowest pathwidth construction all along.
That Haken-Luby is oriented gives the bonus that the resulting fitness landscapes are semismooth -- something that was not shown for the other constructions~\cite{tw7,pw4}.
Very recently, Kaznatcheev and Vazquez Alferez~\textcite{pw2} and van Marle~\textcite{pw2MSc} produced a construction similar to Haken and Luby~\textcite{hakenSteepest} that reduced the bound to pathwidth $2$, 
and showed that their construction is an oriented \VCSP.
This is the lowest possible pathwidth with exponential steepest ascent because all ascents are quadratic for tree-structured \VCSP{s}~\cite{repCP}.
So some of the simplest oriented \VCSP{s} are already hard for greedy local search.

Now, we prove that conditionally-smooth landscape are intractable for \lang{RandomFacet}~\cite{randomFacetBound}, by showing that conditionally-smooth landscapes can express Matou\v{s}ek AUSOs~\cite{randomFacetBound, M94}.

\begin{definition}[G{\"a}rtner \textcite{randomFacetBound}, Matou{\v{s}}ek \textcite{M94}]
Given any $n$ parity functions $P_i: \{0,1\}^{R_i} \rightarrow \{0,1\}$
with scopes such that $i \in R_i \subseteq [i]$,
a landscape on $\{0,1\}^n$ with fitness function
$f(x) = -\sum_{i = 1}^n 2^{n - i} P_{i}(x[R_i])$
is a \emph{Matou\v{s}ek AUSO}.
\label{def:Matousek}
\end{definition}

AUSOs are often defined just in terms of their outmap ($\phi^+$), without specific fitness values, so we had to make a specific choice in \cref{def:Matousek}.
Our choice of fitness function, however, is the simplest one in terms of overall arity that can can implement the out-map of the Matou\v{s}ek AUSOs.
Specifically, Batman~\textcite{arityBSc} showed that any \VCSP{} implementing the same out-map as Matou\v{s}ek AUSOs must have non-zero constraints with scopes $R_i$.
Overall, the fitness function of a Matousek AUSO is very well behaved, in particular:

\begin{proposition}
If $f$ is a Matou\v{s}ek AUSO on $n$ bits then
for all $k \in [n]$ and $y \in \{0,1\}^{[k - 1]}$ there is a preferred assignment $b \in \{0,1\}$ such that $\forall z \in \{0,1\}^{[n] - [k]} \; f(ybz) > f(y\overline{b}z)$.
\label{prop:MatousekCombed}
\end{proposition}

\begin{proof}
This follows from rewriting the big sum for $f$ in 3 terms:
\begin{equation}
f(xby)=
    -\Big({\scriptstyle 2^{n - (k - 1)}\sum_{i=1}^{k-1}2^{(k - 1) - i}P_i(y[R_i]))}\Big) 
    - 2^{n - k}P_k(yb[R_k])  
    - \Big({\scriptstyle 2^{n - k}\sum_{i=1}^{n - k}2^{- i}P_i(ybz[R_i])}\Big)
\end{equation}
and noting that the first summand is independent of $b$ and $z$,
the last summand can sum to at most $-(2^{n - k} - 1)$,
and that the middle summand is independent of $z$ and saves us $2^{n - k}$ in the sum if we set $b = \overline{\text{Parity}(y[R_k - \{k\}])}$.
\end{proof}

From this, it follows that:

\begin{corollary}
Matou\v{s}ek AUSOs are both conditionally-smooth and semismooth.
\label{cor:MatousekPrecSmooth}
\end{corollary}

\noindent Given the similarity of \cref{prop:MatousekCombed} and \cref{prop:orientedVCSPcombed}, if we generalized the definition of oriented binary \VCSP{s} to arbitrary arity then Matou\v{s}ek AUSOs would be oriented.
The real power of \cref{cor:MatousekPrecSmooth} is that we can combine it with the super-polynomial lower-bound on the runtime of \lang{RandomFacet} (Theorem 4.2 from G{\"a}rtner~\textcite{randomFacetBound}) to give:

\begin{theorem}[G{\"a}rtner \textcite{randomFacetBound}]
There exist families of fitness landscapes that are both conditionally-smooth and semismooth such that \lang{RandomFacet} follows an ascent with expected length $\text{exp}(\Theta(\sqrt{n}))$.
\label{cor:RandomFacetSlow}
\end{theorem}

\noindent Given that \lang{RandomFacet} is the best known algorithm for semismooth fitness landscapes, 
it is surprising to see it performing at its worst-case on conditionally-smooth landscapes.

\section{Conclusion and Future Work}
\label{sec:discussion}

Overall, we showed that conditionally-smooth fitness landscapes -- a polynomial-time testable structural property of \VCSP{s} that generalizes the natural notion of oriented \VCSP{s} -- are an expressive subclass of Boolean \VCSP s for which many (but not all) popular local search algorithm are both effective and efficient.
Future work could further generalize conditionally-smooth landscapes from Boolean to higher-valence domains.
This might allow us to engage with directed \VCSP s 
where each strongly connected component is of small size $D \in O(\log n)$ 
by modeling the whole connected component as one domain with $2^D$ values.
Would this class still be efficient-for-many local search algorithms?
Is local search fixed-parameter tractable when parameterized by the size of the \VCSP{'s} largest strongly connected component?

We hope that our results on conditionally-smooth fitness landscapes contribute to a fuller understanding of when local search is both effective and efficient.
This is interesting for theory as it grows our understanding of parameterized complexity, especially for \PLS.
This theory can guide us to what properties practical fitness landscapes might have in the real-world cases where local search seems to work well.
Finally, we are excited about the use of these results for theory-building in the natural sciences.
When local search is used by nature, we often do not have perfect understanding of which local search algorithm nature follows.
But we might have strong beliefs about nature's algorithm being efficient.
In this case, a good understanding of what landscapes are efficient-for-many local search algorithms can help us to reduce the set of fitness landscapes that we consider theoretically possible in nature.
This can be especially useful in fields like evolutionary biology~\cite{evoPLS,KazThesis} and economics~\cite{AlgoEconSurvey} where we have only limited empirical measurements of nature's fitness landscapes.


\newpage
\appendix
\clearpage
\markboth{APPENDICES}{APPENDICES}
\section{Properties of directed and oriented \VCSP s}
\label{sec:propVCSPs}

\begin{proposition}
    Let $\mathcal{C}$ be a binary Boolean \VCSP-instance that implements a fitness landscape $f$ then
    $\mathcal{C} \text{ is directed } \iff  f \text{ is semismooth}$.
    \label{prop:directedVCSP_semismooth}
\end{proposition}

\begin{proof}
    ($\Rightarrow$) By definition $\mathcal{C}$ has no bidirectional arcs. 
    Thus, for any pair of indexes $i\neq j$ we have $\nrse{i}{j}$, 
    and by \cref{prop:semismooth} the fitness landscape implemented by $\mathcal{C}$ is semismooth.
    
\noindent ($\Leftarrow$) By definition, if a fitness landscape is semismooth, its fitness function is such that for all $i\neq j$ it holds that $\nrse{i}{j}$ in all backgrounds. 
    Therefore, there are no bidirected arcs in $\mathcal{C}$.
\end{proof}

Just as the constraints imply the direction of the arcs, the presence or absence of certain arcs implies something about the constraints that generated them.
For example, let us write $j\not\leftarrow k$ if $\arc[b]{j}{k} \not\in A(\mathcal{C})$ and $\arc[l]{j}{k} \not\in A(\mathcal{C})$ then we have the following proposition:

\begin{proposition}    \label{lem:dipath_cij}
    In binary Boolean \VCSP-instance $\mathcal{C}$, if $\se{i}{j}$ and $j \not\leftarrow k$ then $|c_{ij}| > |c_{jk}|$.
\end{proposition}

\begin{proof}
    Since $\se{i}{j} \in A(\mathcal{C})$ 
    and $\hat{c}_j(x,\{i,j\})=\hat{c}_j(x,\{i,j,k\})+c_{jk}x_k$ we have $\text{sgn}(\hat{c}_j(x,\{i,k\})+c_{jk}x_k)\neq \text{sgn}(c_{ij})$ and $|c_{ij}|>|\hat{c}_j(x,\{i,k\})+c_{jk}x_k|$.
    Since $\nse{k}{j}$ we know that for all assignments $y$ we have $\text{sgn}(\hat{c}_j(y,\{k\}))=\text{sgn}(\hat{c}_j(y,\{k\})+ c_{jk})$.
     In particular, $x$ and $x[i\to \Bar{x}_i]$ must be such assignments.
     Thus, $\text{sgn}(\hat{c}_j(x,\{i,k\})) = \text{sgn}(\hat{c}_j(x,\{i,k\}) + c_{jk})
     \neq
     \text{sgn}(\hat{c}_j(x,\{i,k\}) + c_{ij} ) = 
                 \text{sgn}(\hat{c}_j(x,\{i,k\}) +c_{ij}+c_{jk})$.
     If $\text{sgn}(c_{ij})=\text{sgn}(\hat{c}_j(x,\{i,k\}))$ then $|c_{ij}|>|c_{jk}|$.
     In the other case, $\text{sgn}(c_{ij})\neq \text{sgn}(\hat{c}_j(x,\{i,k\}))$ from which $|c_{ij}|>|\hat{c}_j(x,\{i,k\})|>|c_{jk}|$.
\end{proof}

It is important to notice that we do not have the freedom to assign the direction of arcs arbitrarily.
Rather, they are determined by the constraint weights of the \VCSP-instance.
The overall structure of the network is closely linked to the direction of the arcs:
In particular, to have a lot of arcs in opposite directions, a directed \VCSP-instance needs to have a highly clustered network.
Or, more precisely, in the contrapositive:

\begin{theorem}\label{lem:triangle-free-means-oriented}
    If a directed \VCSP-instance $\mathcal{C}$ is triangle-free then $\mathcal{C}$ is oriented.
\end{theorem}

\noindent This theorem matters for two reasons.
First, it shows that oriented \VCSP{s} are in some sense a generalization of the tree-structured \VCSP{s} studied by Kaznatcheev, Cohen, and Jeavons~\textcite{repCP},
since tree-structured directed \VCSP{s} are triangle-free (and thus oriented by \cref{lem:triangle-free-means-oriented}).
This means that our results in \cref{app:efficientAlgos} showing that oriented \VCSP{s} are efficient-for-many local search algorithms are an extension of Kaznatcheev, Cohen, and Jeavons~\textcite{repCP}'s results that tree-structured binary Boolean \VCSP{s} are efficient-for-all local search algorithms.
Second, \cref{lem:triangle-free-means-oriented} suggests that oriented \VCSP{s} are not ``rare'' -- a large natural subclass of directed \VCSP{s} is oriented.
We prove \cref{lem:triangle-free-means-oriented} by contrapositive from a more detailed proposition:

\begin{proposition}\label{prop:unclustered_bi}
Given a fitness landscape $f$ implemented by \VCSP-instance $\mathcal{C}$ with neighbourhood $N$,
if $N(i) \cap N(j)$ is empty and both $\se[y]{i}{j}$ and $\se[x]{j}{i}$ in $f$ then $\arc[b]{i}{j} \in A(\mathcal{C})$.
\end{proposition}

\begin{proof}        
Given that $N(i) \cap N(j)$ is empty, define an assignment $z$ as $z_k = x_k$ for $k \in N(i)$, $z_k = y_k$ for $k \in N(j)$, and $z_k = 0$ otherwise.
From \cref{def:effective-unary}, we have that:
$\hat{c}_i(z,\{j\})=\hat{c}_i(x,\{j\})$ and $\hat{c}_j(z,\{i\})=\hat{c}_j(y,\{i\})$.
From $\se[x]{j}{i}$ and \cref{def:VCSP_arcs} we have $|c_{ij}|>|\hat{c}_i(x,\{j\})|$ and $\text{sgn}(c_{ij})\neq \text{sgn}(\hat{c}_i(x,\{j\}))$.
Similarly, from $\se[y]{i}{j}$ and \cref{def:VCSP_arcs} we know $|c_{ij}|>|\hat{c}_j(y,\{i\})|$ and $\text{sgn}(c_{ij})\neq \text{sgn}(\hat{c}_j(y,\{i\}))$.
It follows that $|c_{ij}|>|\hat{c}_i(z,\{j\})|$ with $\text{sgn}(c_{ij})\neq \text{sgn}(\hat{c}_i(z,\{j\}))$
and $|c_{ij}|>|\hat{c}_j(z,\{i\})|$ with $\text{sgn}(c_{ij})\neq \text{sgn}(\hat{c}_j(z,\{i\}))$.  
Thus $z$ is a certificate that $\arc[b]{i}{j}\in A(\mathcal{C})$.
\end{proof}

\begin{proof}[Proof of \cref{lem:triangle-free-means-oriented}]
In a triangle-free graph, $N(i) \cap N(j)$ is empty if $i \neq j$. 
Since $\mathcal{C}$ is directed, $\arc[b]{i}{j} \not\in A(\mathcal{C})$.
By \cref{def:VCSP_arcs,prop:unclustered_bi} either $i \rightarrow j \not\in A(\mathcal{C})$ or $j \rightarrow i \not\in A(\mathcal{C})$.
\end{proof}

The relationship between arc direction and overall structure of the graph is not just local:

\begin{proposition}\label{prop:no-cycles}
    An oriented binary Boolean \VCSP-instance $\mathcal{C}$ has no directed cycles --
    its constraint graph is a directed acyclic graph (DAG).
\end{proposition}

\begin{proof}
    Take any directed walk $W= a \to b \to \ldots \to i \to j \to k \to \ldots \to r \to s$ in $\mathcal{C}$. 
    Since $\mathcal{C}$ is oriented and $\arc[r]{j}{k} \in A(\mathcal{C})$ that means $j \not\leftarrow k$.
    Thus, by \cref{lem:dipath_cij} the magnitude of the binary constraints decreases along the walk, i.e. $|c_{ab}| > \ldots > |c_{ij}| > |c_{jk}| > \ldots >|c_{rs}|$.
    So $W$ has no cycles.
\end{proof}

\noindent Much like \cref{lem:triangle-free-means-oriented}, this proposition lets us views oriented \VCSP{s} as a kind of generalization or extension of tree-structured \VCSP{s}: moving from undirected acyclicity of trees to the directed acyclicity of DAGs.
That oriented \VCSP s have acyclic constraint graphs (\cref{prop:no-cycles}) can be also combined with the concept of a \emph{combed face}:

\begin{definition}
Given any set $R \subseteq [n]$, we will say that an assignment $x$ is in the face $\{0,1\}^R y$ with background partial assignment $y \in \{0,1\}^{[n] \setminus R}$ (and write $x \in \{0,1\}^R y$) if $x[[n] \setminus R] = y$.
We will say that a face $\{0,1\}^R y$ is \emph{combed} along dimension $i \in R$ with preferred assignment $x^*_i(y)$ if for all $x \in \{0,1\}^R y$ we have $f(x[i \mapsto x^*_i(y)]) > f(x[i \mapsto \overline{x^*_i(y)}])$.
\label{def:combed}
\end{definition}

\noindent In other words, a face is combed if there exists a dimension such that all fitness increasing steps point in the same direction along that dimension -- as if it was hair that was combed.
This gives us a combing property of oriented \VCSP{s}:

\begin{proposition}
If strict ordered set $([n],\prec)$ is a transitive closure of the oriented \VCSP{'s} constraint DAG then for each $j \in [n]$ with $\downarrow j = \{i \; |\; i \prec j\}$, $\uparrow j = \{k \; | \; j \preceq k \}$ and any partial assignment $y \in \{0,1\}^{\downarrow j}$ the face $\{0,1\}^{\uparrow j} y$ is combed along $j$. 
\label{prop:orientedVCSPcombed}
\end{proposition}

\begin{proof}
    Any arc between $j$ and $k$ for $k\in\uparrow j$ is of the form $\se{j}{k}$.
    So the variable $x_j$ does not sign depend on any $x_k$ with $k\in\uparrow j$, and $f(x[j \mapsto x^*_j(y)]) > f(x[j \mapsto \overline{x^*_j(y)}])$ for any $x\in\{0,1\}^{\uparrow j}$.
\end{proof}

\noindent Thus, the preferred assignment of a variable $x_j$ depends only on the assignments of variables $x_i$ with $i$ lower than $j$ in the partial order (i.e., $i \prec j$).
The preferred assignment of $x_j$ is conditionally independent of all other variables conditional on variables $x_i$ with $i \prec j$ being fixed to any value. 
This is closely related to recursively combed AUSOs~\cite{GGL20,randomFacetBound,M06}:

\begin{definition}[\textcite{GGL20}]
A semismooth fitness landscape $f$ on $\{0,1\}^n$ is \emph{recursively combed} if:
(a) $f$ is a smooth landscape, or (b) $f$ is combed along some dimension $i \in [n]$ and both faces $\{0,1\}^{[n]\setminus \{i\}} 0$ and $\{0,1\}^{[n]\setminus \{i\}} 1$ are recursively combed.
\label{def:rcAUSO}
\end{definition}

\noindent The only distinction between the above definition versus landscapes implemented by oriented \VCSP s is that the combing of the faces $\{0,1\}^{[n]\setminus \{i\}} 0$ and $\{0,1\}^{[n]\setminus \{i\}} 1$ might be along different dimensions in a recursively combed AUSO but have to be combed along a common set of dimensions for oriented \VCSP s.

For the efficiency of local search, however, the above distinction 
is without a difference.
We can abstract beyond recursively combed AUSOs by loosening Definition~\ref{def:rcAUSO} by replacing the 
recursive condition (b) that both faces $\{0,1\}^{[n]\setminus \{i\}} 0$ and $\{0,1\}^{[n]\setminus \{i\}} 1$ are recursed on 
by the condition that only $\{0,1\}^{[n]\setminus \{i\}} x^*_i$ is recursed on -- 
where $x^*_i$ is the preferred assignment of the combed dimension $i$.
We saw this stated iteratively as \cref{def:combed}.
\section{Recognizing various subclasses of \VCSP{s} and fitness landscapes}
\label{sec:recognitionAlgs}

Subclasses of \VCSP s are more useful when there is a recognition algorithm that can answer for any specific \VCSP-instance if that instance is in the (sub)class or not.
In this section, we provide recognition algorithms for \VCSP s that implement smooth landscapes, directed and oriented \VCSP s, and \VCSP s that implement $\prec$-smooth landscapes.
The first and last of these algorithms run in polynomial time.
But the algorithm for recognizing directed and oriented \VCSP s, or -- more generally -- the algorithm for determing the set of arcs of a \VCSP-instance $\mathcal{C}$ has parameterized complexity of $O(2^{\Delta(\mathcal{C})}|\mathcal{C}|)$ where $\Delta(\mathcal{C})$ is the max degree of the constraint graph of $\mathcal{C}$.
We argue that this is the best possible through a pair of reductions from \lang{SubsetSum}.

All the algorithms in this section will rely on computing the effective unary $\hat{c}_i(x,R)$ from \cref{def:effective-unary}.
If all our constraint weights are in $[-2^l,2^l]$ then this can be done in time $O(|R||\mathcal{C}|l)$.
But we will not worry about these implementation specific details and treat the computation of $\hat{c}_i(x,R)$ as one time-step.
Similarly with other implementation-specific operations like seeing if $c_{ij} \in \mathcal{C}$ or checking if $i \in N_\mathcal{C}(j)$ and basic arithmetic operations on weights and indexes -- we will treat all of these as one timestep.

\subsection{Recognizing \VCSP s that implement smooth landscapes}
\label{sec:smooth_rec}

For recognizing smooth landscapes we will describe a slightly more general algorithm.
The algorithm \lang{ConditionallySignIndependent}$(\mathcal{C},i,S,y)$ will check if $i$ is conditionally sign independent once the variables with indexes in $S$ have their values set to $y[S]$.
This is equivalent to checking if $i$ is sign-independent in the face $\{0,1\}^{[n] \setminus S}y$ or checking if the face $\{0,1\}^{[n] \setminus S}y$ is combed along dimension $i$ (recall \cref{def:combed}).

For a variable $x_i$ to be sign-independent, its unary must dominate over the binary constraints that it participates in -- its unary must have big magnitude.
For example, suppose that $c_i > 0$ then $x_i$ will prefer to be $1$ if all of its neighbours are $0$.
For $x_i$ to be sign independent, this preference must not change as any $x_j$ with $j \in N(i)$ change to $1$s.
Since $x_j$ with $c_{ij} < 0$ are the only ones that can lower $x_i$'s preference for $1$, this becomes equivalent to checking if
$c_i > \sum_{j \in N(i) \text{ s.t. } c_{ij} < 0}|c_{ij}|$.
There exists a combination of $x_j$s that flip $x_i$'s preference if and only if this inequality is violated.
To move from this particular example to the general case, we only need to replace $c_i$ by $\hat{c}_i = \hat{c}_i(y,N(i)\backslash S)$ and consider $c_{ij}$s with $\text{sgn}(\hat{c}_i) \neq \text{sgn}(c_{ij})$ and $j \in N(i) \setminus S$.
We encode this is \cref{alg:sign-independence}.

\printAlgSI*

This algorithm runs in linear time or even just $\Delta(\mathcal{C})$ time-steps.
To check if a landscape is smooth, run \lang{ConditionallySignIndependent}$(\mathcal{C},i,\emptyset,0^n)$ on every $i$ and return the conjunction of their outputs (in a total of $\Delta(\mathcal{C})n$ time-steps).

\subsection{Recognizing directed and oriented \VCSP s}
\label{sec:complexity_arcs}

At the highest level, the overarching idea for checking if a \VCSP\ is directed or oriented can feel like it is not too fundamentally different from the check for smoothness.
Instead of checking each variable, however, we need to check each edge to determine its `direction'.
This is straightforward given that the definition of arcs (\cref{def:VCSP_arcs}) is already expressed in terms of a comparison of $|\hat{c}_i(x,\{i,j\})|$ and $|c_{ij}|$.
The issue arises in that instead of finding a big combination of $|c_{ij}|$s -- as in the smoothness check -- here, we need to find an $x$ that makes $|\hat{c}_i(x,\{i,j\})|$ small.
And this flip of direction makes all the difference for the computational complexity.

For example, if we want to check the potential direction of an edge $\{i,j\} \in E(\mathcal{C})$ with $c_{ij} > 0$ then we need to find an argument $x$ such that $|\hat{c}_j(x,\{i,j\})|$ is minimized while keeping $\hat{c}_j(x,\{i,j\}) \leq 0$.
Now, if we want to check if the edge should be directed as $\arc[r]{i}{j}$ then we just need to check if $c_{ij} > |\hat{c}_j(x,\{i,j\})|$.
Unfortunately, this minimization over $x$ cannot be done quickly in the general case since it is hiding \lang{SubsetSum} within it (as we will see in \cref{prop:rse_NPC} and \cref{prop:doublese_NPC}).
Some speed up is possible, however, since we do not need to vary candidate assignments over all of $\{0,1\}^n$, but can just focus on $x \in \{0,1\}^{N(i) \setminus \{j\}}0^{n + 1 - |N(i)|}$ 
(i.e., over the face spanned by the neighbours of $i$) 
since $\hat{c}_j(x,\{i,j\}) = \hat{c}_j(y,\{i,j\})$ if $x[N(i) \setminus \{j\}] = y[N(i) \setminus \{j\}]$.
It is this observation that will enable us to get a complexity parameterized by the max degree $\Delta(\mathcal{C}) \geq |N(i)|$.
With this intuition in mind, we can give the algorithm that implements the definition of arc directions (\cref{def:VCSP_arcs}) as \cref{alg:VCSP_arcs_min}.

\begin{algorithm}[h]
\caption{\lang{ArcDirection}($\mathcal{C},i,j$):}
\label{alg:VCSP_arcs_min}
\begin{algorithmic}[1]
\Require \VCSP-instance $\mathcal{C}$ and variable indexes $i$ and $j$.
\State Set $\hat{c} \gets \min_{x \in \{0,1\}^{N(i) \cup N(j) \setminus \{i,j\}}} \max_{k \in \{i,j\}} |\hat{c}_k(x,\{i,j\})|$ s.t. $\text{sgn}(\hat{c}_i(x,\{i,j\}) \neq \text{sgn}(c_{ij})$ and $\text{sgn}(\hat{c}_j(x,\{i,j\}) \neq \text{sgn}(c_{ij})$
\State \algorithmicif\ $\hat{c} < |c_{ij}|$ \algorithmicthen\ \textbf{return} $\{\arc[b]{i}{j}\}$
\State $A \gets \{\}$
\State Set $\hat{c}_i \gets \min_{y \in \{0,1\}^{N(i) \setminus \{j\}}} |\hat{c}_i(y,\{i,j\})|$ s.t. $\text{sgn}(\hat{c}_i(y,\{i,j\}) \neq \text{sgn}(c_{ij})$
\State Set $\hat{c}_j \gets \min_{z \in \{0,1\}^{N(j) \setminus \{i\}}} |\hat{c}_j(z,\{i,j\})|$ s.t. $\text{sgn}(\hat{c}_j(z,\{i,j\}) \neq \text{sgn}(c_{ij})$
\State \algorithmicif\ $\hat{c}_i < |c_{ij}|$ \algorithmicthen\ $A \gets A \cup \{\arc[r]{i}{j}\}$
\State \algorithmicif\ $\hat{c}_j < |c_{ij}|$ \algorithmicthen\ $A \gets A \cup \{\arc[l]{i}{j}\}$
\State \textbf{return} $A$
\end{algorithmic}
\end{algorithm}

Given the brute force minimization for $\hat{c}$, $\hat{c}_i$, and $\hat{c}_j$ in \cref{alg:VCSP_arcs_min}, the resulting runtime is $2^{|N(i)| - 1} + 2^{|N(j)| - 1} + 2^{|N(i)| + |N(j)| - 2}$ time-steps or $O(2^{2(\Delta(\mathcal{C}) - 1)})$ time-steps in the worst-case.
To find the whole set of arcs -- and thus be able to decide if $\mathcal{C}$ is directed or oriented or not -- we need to repeat \lang{ArcDirection}($\mathcal{C},i,j$) for every $\{i,j\} \in E(\mathcal{C})$, requiring $O(2^{\Delta(\mathcal{C})}|\mathcal{C}|)$ time-steps.
Unfortunately, finding the direction of the arcs of a \VCSP\ cannot be done significantly faster than \cref{alg:VCSP_arcs} because determining the direction of arcs is \NP-hard by reduction from \lang{SubsetSum}.
\begin{figure}
    \centering
    \begin{subfigure}[b]{0.49\textwidth}
    \begin{tikzpicture}
        \node[draw, circle, minimum size=4pt, label=$3$] (center) at (0,0) {$c$};
        
        \foreach \x [count = \xi] in {1,2}{
            \node[draw, circle, minimum size=4pt, label=left:$1$] (outer\xi) at (-2.5,\xi-2) {$\x$};
            
            \draw (center) -- (outer\xi) node [near end,above,sloped, fill=white] (w\x){$4a_\x$};
        
            }
        \node[draw, circle, minimum size=4pt, label=left:$1$] (outern) at (-2.5, 2) {$n$};
        
        \draw (center) -- (outern) node [near end,above,sloped, fill=white] (wn){$4a_n$};
        
        \path (outer2) -- (outern) node [midway, fill=white] {$\vdots$};
        
        \node[draw, circle, minimum size=4pt, label=$1$] (outerl) at (2,0) {$l$};
        \node[draw, circle, minimum size=4pt, label=below:$-1$] (outer0) at (0,-2.5) {$0$};
        \draw (center) -- (outer0) node [midway, above, sloped, fill=white] {$-4T-2$};
        \draw (center) -- (outerl) node [midway, above, sloped, fill=white] {$-2$};
    \end{tikzpicture}
    \subcaption{\VCSP-instance for \cref{prop:rse_NPC}}
    \label{fig:reductionDirected}
    \end{subfigure}
    \hfill
    \begin{subfigure}[b]{0.49\textwidth}
        \begin{tikzpicture}
        \node[draw, circle, minimum size=4pt, label=left:$3$] (center) at (0,0) {$c$};
        
        \node[draw, circle, minimum size=4pt, label=right:$1$] (outerl) at (3,0) {$l$};
        
        \node[draw, circle, minimum size=4pt, label=above:$M+1$] (outer1) at (0, 2) {$1$};
        \draw (center) -- (outer1) node [midway,above,sloped, fill=white] (wn){$4a_1$};
        \draw (outerl) -- (outer1) node [near end,above,sloped, fill=white] (wn){$-4a_1$};

        \node[draw, circle, minimum size=4pt, label=above:$M+1$] (outern) at (3, 2) {$n$};
        \draw (center) -- (outern) node [near end,above,sloped, fill=white] (wn){$4a_n$};
        \draw (outerl) -- (outern) node [midway,above,sloped, fill=white] (wn){$-4a_n$};
        
        \path (outer1) -- (outern) node [midway, fill=white] {$\cdots$};

        \node[draw, circle, minimum size=4pt, label=below:$ M +1$] (outer0) at (1.5,-2.5) {$0$};
        \draw (outer0) -- (center) node [midway, below, sloped, fill=white] {$-4T-2$};
        \draw (outerl) -- (outer0) node [midway, above, sloped, fill=white] {$4T+2$};
        \draw (center) -- (outerl) node [midway, above, sloped, fill=white] {$-2$};
    \end{tikzpicture}
    \subcaption{\VCSP-instance for \cref{prop:doublese_NPC}}
    \label{fig:reductionOriented}
    \end{subfigure}
    \caption{The \VCSP\ instance $\mathcal{C}$ in the reductions from positive \lang{SubsetSum} of \cref{prop:rse_NPC} and \cref{prop:doublese_NPC}. 
    The weight of the unary constraints is given near the nodes, whilst the weight of the binary constraints appears on the edges.
    In both (a) and (b) the variables $0,1,\ldots, n$ are sign-independent.
    In (a) the \lang{SubsetSum} instance is encoded in the binary constraints of variables $0,1,\ldots n,c$, and $\ell$ sign-depends on $c$. Only the edge $\{c,l\}$ can be bidirected, and whether or not it is depends on the answer to the \lang{SubsetSum} instance.
    In (b) the encoding of \lang{SubsetSum} is also on the binary constraints, but now $l$ sign-depends on $c$ only when $x_0=x_1=\ldots=x_n=0$, and $c$ sign depends on $l$ if and only if \lang{SubsetSum} is a yes-instance.
    }
\end{figure}
\begin{proposition}
Given a \VCSP\ instance $\mathcal{C}$ and a pair of indexes $i$ and $j$ the question ``is $\rse{i}{j} \in A(\mathcal{C})$?" is \NP-complete.
\label{prop:rse_NPC}
\end{proposition}

\begin{proof}
It is clear that this problem is in \NP{} as a background assignment $x$ with reciprocal sign epistasis for $i$ and $j$ acts as a certificate. 
We will show \NP-completeness by reduction from positive \lang{SubsetSum}.
Consider an instance $\mathcal{S}$ of positive \lang{SubsetSum} consisting of $n$ positive integers $a_1,\ldots,a_n$ and a positive integer target-sum $T$.
Based on this, we will create a \VCSP-instance $\mathcal{C}$ on $n + 3$ variables with indexes given by $\{c,l,0,1,\ldots,n\}$.
Its constraint graph will be a star with $c$ at the center.
The unary constraints will be given by $w_c = 3$ and $w_l = w_1 = \ldots = w_n = 1$ and $w_0 = -1$.
Binary constraints will be given by $w_{cl} = -2$, $w_{c0} = -4T - 2$, and $w_{ci} = 4a_i$.
See \cref{fig:reductionDirected} for illustration.
Looking at any partial assignment to $x_0,x_1,\ldots x_n$ we can see that $\se{c}{l}$.
Our proposition follows from the observation that:
$\arc[b]{c}{l} \in A(\mathcal{C}) \iff \exists S \subseteq [n] \text{ s.t } \sum_{i \in S} a_i = T$
\end{proof}

This question remains difficult even if we have a directed \VCSP-instance and want to know if a pair of variables have arcs in both directions.

\begin{proposition}
Given a \VCSP\ instance $\mathcal{C}$ and a pair of indexes $i$ and $j$ the question ``are both $\se{i}{j}$ and $\se{j}{i}$ in $A(\mathcal{C})$?" is \NP-complete.
\label{prop:doublese_NPC}
\end{proposition}

The proof is similar to \cref{prop:rse_NPC}, with slight changes to the resulting \VCSP-instance.

\begin{proof}
For membership in \NP{} a pair of assignments $x,y$, with $\se{i}{j}$ and $\se{j}{i}$ respectively, act as the certificate.
As with the proof of \cref{prop:rse_NPC}, we will show \NP-hardness by reduction from positive \lang{SubsetSum}.
Consider an instance $\mathcal{S}$ of positive \lang{SubsetSum} consisting of $n$ positive integers $a_1,\ldots,a_n$ and a positive integer target-sum $T$.
Let $M = T + \sum_{i = 1}^n a_i$.
Based on this, we will create a \VCSP-instance $\mathcal{C}'$ on $n + 3$ variables with indexes given by $\{c,l,0,1,\ldots,n\}$.
Its constraint graph will be a complete split graph with a clique $\{c,l\}$ and independent set $\{0,1,\ldots,n\}$.
The unary constraints will be given by $w_c = 3$, $w_l = 1$ and $w_0 = w_1 = \ldots = w_n = M + 1$.
Binary constraints will be given by $w_{cl} = -2$, $-w_{c0} =  w_{l0} = 4T + 2$, and $-w_{ci} = w_{li} = -4a_i$.
See \cref{fig:reductionOriented} for illustration.
If we look at the partial assignment $x_0 = x_1 = \ldots = x_n = 0$, we can see that $\arc[r]{c}{l}\in A(\mathcal{C'})$.
Our proposition follows from the observation that:
$\arc[l]{c}{l} \in A(\mathcal{C'}) \iff \exists S \subseteq [n] \text{ s.t } \sum_{i \in S} a_i = T$
\end{proof}

Given that subexponential time algorithms for subset sum are not known,
it makes sense to simplify \cref{alg:VCSP_arcs_min} by replacing the minimizations by brute-force search over the neighbours of $i$ and $j$.
We use this variant the main text (\cref{sec:recognize_doVCSP}) and restate it here as \cref{alg:VCSP_arcs}:

\printAlgArcs*

The individual edge hardness reuslts of \cref{prop:rse_NPC,prop:doublese_NPC} also give us corollaries about the difficulty of checking if a \VCSP-instance is directed or oriented.
The complexity of checking either oriented or directed follows from \cref{prop:rse_NPC}:

\begin{corollary}
Given a \VCSP\ instance $\mathcal{C}$ the questions of ``is $\mathcal{C}$ directed?" or ``is $\mathcal{C}$ oriented?" are both \coNP-complete.
\label{cor:semismoothTesting_coNPC}
\end{corollary}
\begin{proof}
    In the construction of \cref{prop:rse_NPC}, the variables $0,1,\ldots, n$ are all sign independent as the unary and binary constraints have the same sign. 
    Therefore the only edge that has the potential to be bidirected is $\{c,\ell\}$, for which we know $\se{c}{\ell}$ and \lang{SubsetSum}($\mathcal{S}$) establishes the other direction:
    If \lang{SubsetSum}($\mathcal{S}$) is false, then $\mathcal{C}$ is oriented (and therefore directed). 
    If \lang{SubsetSum}($\mathcal{S}$) is true, then $\mathcal{C}$ is not even directed as it contains $\rse{c}{\ell}$.
    Therefore, the same construction (and the same certificate) shows that the questions "is $\mathcal{C}$ not directed?" and "is $\mathcal{C}$ not oriented?" are both \NP-complete.
    This establishes the \coNP-completeness of the complement languages "is $\mathcal{C}$ directed?" and "is $\mathcal{C}$ oriented?".
\end{proof}

Since the constructed \VCSP-instance $\mathcal{C}'$ in the proof of \cref{prop:doublese_NPC} is oriented if the \lang{SubsetSum}($\mathcal{S}$) is false and directed (but not oriented) if \lang{SubsetSum}($\mathcal{S}$) is true, the complexity of checking if a directed \VCSP-instance is oriented follows:

\begin{corollary}
    For directed \VCSP-instance $\mathcal{C}$ the question ``is $\mathcal{C}$ oriented?" is \coNP-complete.
    \label{cor:orientedTesting_coNPC}
\end{corollary}
\begin{proof}
    In the construction of \cref{prop:doublese_NPC} all variables with indexes in $[n]\cup \{0\}$ are sign-independent.
    Therefore only the edge $\{c,l\}$, for which we know $\se{c}{l}$ has the potential to be directed in both directions. 
    Note that $\nrse{c}{l}$ since we have $\se{c}{l}$ only when $x_0=x_1=\ldots=x_n=0$ but under this partial assignment $\nse{l}{c}$.
    \lang{SubsetSum} determines whether some other partial assignment exists for which $\se{l}{c}$.
\end{proof}

\subsection{Recognizing \VCSP s that implement conditionally-smooth landscapes}
\label{sec:precSmooth_rec}

As we generalize from oriented to conditionally-smooth, the recognition task becomes easier.
The overall approach is similar to checking if a \VCSP\ is smooth (i.e., $\emptyset$-smooth) with the difference being that subsequent calls to \lang{ConditionallySignIndependent}$(\mathcal{C},i,S,y)$ adjust the set $S$ and background assignment $y$ based on previous calls.
This gives the \lang{ConditionallySmooth}$(\mathcal{C})$ algorithm (\cref{alg:precSmoothCheck}).

\printAlgPrecSmooth*

To analyze \cref{alg:precSmoothCheck} it is helpful to have a recursive variant of the iterative definition of conditionally smooth given in \cref{def:precSmooth}.
For this, we note:

\begin{proposition}\label{prop:precsmooth_recursive}
Given a conditionally-smooth fitness landscape $f$ there must be a non-empty set $S \subseteq [n]$ such that for all $i \in S$ $i$ is sign-independent in $f$ with preferred asignment $x^*_i$.
Further, given any such set $S$, $f$ restricted to $\{0,1\}^{[n] \setminus S}x^*[S]$ must be conditionally-smooth.
\end{proposition}

By using the observation that smooth landscapes are conditionally-smooth as the base case, \cref{prop:precsmooth_recursive} provides us with an alternative recursive definition of conditionally-smooth landscapes.

Returning to \cref{alg:precSmoothCheck}, notice that if $\mathcal{C}$ implements a conditionally-smooth landscape then by \cref{prop:precsmooth_recursive} we have the loop invariant that $\{0,1\}^{[n] \setminus S}x^*[S]$ is a conditionally-smooth landscape and that thus $T$ must be non-empty 
(hence the `\textbf{return} \textsc{False}' condition on \cref{alg:precSmoothCheck_line:false} if $T$ is empty).
Since we must have $T \subseteq [n] \setminus S$ with $|T| \geq 1$ for each $S$, that means $|S|$ will grow by at least one index at each iteration of the while loop.
So the algorithm calls \lang{ConditionallySignIndependent} at most $\binom{n}{2}$ times for an overall runtime of $\binom{n}{2}\Delta(\mathcal{C})$ timesteps.

Notice that \cref{alg:precSmoothCheck} could also be interpreted as a local search algorithm that follows an ascent defined by each bit $x_i$ that is flipped from $0$ to $1$ in repeated iterations of \cref{alg:precSmoothCheck_line:step}.
Since each $x_i$ is flipped at most once and only if it disagrees with the unique local peak $x^*$ of the $\prec$-smooth landscape implemented by $\mathcal{C}$.
That means this ascent is direct and takes a number of steps equals to the Hamming-distance-to-peak $h_f(0^n) = ||x^* - 0^n||_H$.
This means that at least some local search algorithms are efficient for conditionally-smooth landscapes.
As we will see in \cref{app:efficientAlgos}, this is one of many efficient local search algorithms for conditionally-smooth landscapes where the others will work even if they do not have access to the \VCSP-instance $\mathcal{C}$ that implements the landscape.
\section{Efficient local search in conditionally-smooth landscapes}
\label{app:efficientAlgos}

Given a fitness landscape $f$, local search starts at an initial assignment $x^0$ and steps to further assignments $x^1, x^2, \ldots , x^T$ reaching an assignment $x^T$ that is a local peak.
For an arbitrary local search algorithm \lang{A} this can be summarized as \cref{alg:local-search} where $\lang{A}^1_f(x)$ is a single step of \lang{A} from $x$ on $f$ and $\lang{A}^t_f(x)$ is short-hand for calling $\lang{A}^1_f$ $t$-times: $A^t_f(x) = \underbrace{A^1_f(A^1_f(\ldots A^1_f(x) \ldots ))}_{t \text{ times}}$.

\begin{algorithm}[htb]
\caption{Local Search Algorithm \lang{A}$(f,x^0)$:}\label{alg:local-search}
    \begin{algorithmic}[1]
    \Require A fitness landscape $f$, and a starting assignment $x^0$.
    \Ensure A local peak $x^*$.
    \State Initialize $x\gets x^0$.
        \While{$(\phi^+(x)\neq\emptyset)$}
        \State $x\gets \lang{A}_f^1(x)$
        \EndWhile
        \State \textbf{return} $x$
    \end{algorithmic}    
\end{algorithm}

To show that conditionally-smooth fitness landscapes are efficient-for-many local search algorithms our analysis uses a partition of the $n$ indexes of a conditionally-smooth landscape on $n$ dimensions.
But in order to formally define these partitions and relate them to an assignment $x$ we need to refine our definition of in- and out-maps in a way that is specific to conditionally-smooth landscapes:
\defInOutMap*

\noindent Note that 
given any ascent $x^0,x^1, \ldots , x^T$, we have $\phi^\ominus(x^0) \subseteq \phi^\ominus(x^1) \subseteq \ldots \phi^\ominus(x^T) = [n]$.
We refer to the set $\overline{\phi^\ominus(x)} := [n]-\phi^\ominus(x)$ as the \emph{free indices} at $x$.

Given a $\prec$-smooth landscape $f$, we will let the height of $f$ be the $\text{height}(\prec)$ of the poset $([n],\prec)$ and the width of $f$ be the $\text{width}(\prec)$ of the poset $([n],\prec)$.
We also partition $[n]$ into $\text{height}(\prec)$-many \emph{level sets}.
Specifically, given a variable index $i \in [n]$, define the upper set of $i$ as $\uparrow i = \{j \;|\; i \preceq j \}$.
Then, we partition $[n]$ into $\text{height}(\prec)$-many \emph{level sets} defined as $S_l = \{ i \;|\; \text{height}(\uparrow i, \prec) = l\}$.
Additionally we define $S_0=\emptyset$, and we let $S_{<l}=\bigcup_{k=0}^{l-1}S_k$.

For an assignment $x$, define $\text{height}_f(x)$ (and $\text{width}_f(x)$) as the height (and width) of the poset $(\overline{\phi^\ominus(x)},\prec)$.

Note that for $x$ with $\text{height}(x)=l$, we have that for $l' > l$, $S_{l'} \subseteq \phi^\ominus(x)$, $S_l \subseteq \phi^\ominus(x) \cup \phi^\oplus(x)$, and, for $x\neq x^*$, $\phi^\oplus(x) \cap S_l\neq\emptyset$.
In other words, if $x$ is at height $l$ then all the variables with index at level $l' > l$ are set correctly, all free indexes at level $l$ are at the border at $x$, and at least one variable with index at level $l$ is free for $l\neq 0$.
This means that if $\height[x]=l$ then $\overline{\phi^\ominus(x[\phi^\oplus(x) \mapsto \overline{x[\phi^\oplus(x)]}])} \in S_{<l}$.
With this terminology in place, we can state our main theorem of this section:
\mainStepsBound*

\begin{proof}
    This follows by induction on $\height[x^0]$ and linearity of expected value.
\end{proof}

As is exemplified in \cref{sec:steepestAscent}, one of the main pitfalls of local search that leads to exponentially long ascents, is the near-exhaustive exploration of subfaces of the cube that are far from the peak. 
Such exploration corresponds to changing the same subset of variables many times, for each flip of a variable outside the subset. 
Because of the hierarchical structure of conditionally-smooth landscapes, this near-exhaustive search is avoidable. 
There are many ways to evade getting stuck changing the same number of variables; be that via randomness, as is the case of ascent-following or ascent-biased random local search; or by keeping track of variables that have already been changed as is the case for history-based local search; or by flipping many variables at a time as non-adjacent local search does.

\subsection{Ascent-following random local search}
\label{app:efficientWalking}

Given that \cref{thm:main_stepsBound} is expressed in terms of a stochastic process, it makes sense to first apply it to the prototypical stochastic local search algorithm: \lang{RandomAscent}~\cite{simplexSurvey,randomFitter1,randomFitter2}.
Given an assignment $x$, the step $\lang{RandomAscent}^1_f$ simply returns a fitter adjacent assignment uniformly at random.
Or, slightly more formally, if $Y^\text{RA}(x) \sim \lang{RandomAscent}^1_f(x)$ then:

\begin{equation}
\text{Pr}\{Y^\text{RA}(x) = x[i \mapsto \bar{x_i}]\} = \begin{cases}
\frac{1}{|\phi^+(x)|} & \text{ if } i \in \phi^+(x) \\
0 & \text{otherwise}
\end{cases}
\label{eq:RA}
\end{equation}

\noindent Now, we bound the expected number of steps to decrement height:

\begin{lemma}
Let $X_{<l}=\{x| \text{height}(x)<l\}$ and $s_{<l} = \max_{x\in X_{<l}} |\phi^+(x)|$. 
Then 
$p(n,l) = |S_l| + (1 + \log |S_l|) s_{<l}$
satisfies the condition from \cref{eq:drift_time_bound} of \cref{thm:main_stepsBound} for \lang{RandomAscent}.
\label{lem:p_RA}
\end{lemma}

\begin{proof}
Consider any $x$ with $\text{height}(x)=l$, define $z(x) = |S_l \cap \overline{\phi^\ominus(x)}|$.
The expected time to decrease $z(x)$ is $\frac{1}{p^+(x)}$ where $p^+(x) = \frac{|S^l \cap \overline{\phi^\ominus(x)}|}{|\phi^+(x)|} = \frac{z(x)}{z(x) + |\phi^+(x) \setminus S_l|} \geq \frac{z(x)}{z(x) + s_{<l}}$.
From this we get:
\begin{equation}
\mathbb{E}\{\tau_{<l}(\omega)\} \leq \sum_{z = 1}^{|S_l|} \frac{z + s_{<l}}{z} = |S_l| + s_{<l}\sum_{z = 1}^{|S_l|} \frac{1}{z} \leq |S_l| + (1 + \log |S_l|) s_{<l}
\end{equation}
\end{proof}

The bound from \cref{lem:p_RA} on the expected number of steps to decrement height can be fed carefully into \cref{thm:main_stepsBound} to bound the expected total number of steps for \lang{RandomAscent}:
\propRAstepsBound*

\begin{proof}\label{proof:prop:RA_stepsBound}
To get the bound in \cref{eq:RA_x0bound}, note that \lang{RandomAscent} from $x^0$ will be confined to the face $\{0,1\}^{\overline{\phi^\ominus(x^0)}}x^0[\phi^\ominus(x)]$, so we can redefine all $S_l$s with respect to that face.
Under that definition, $h = \height[x^0]$ and $|S_1| + \ldots + |S_h| = |\overline{\phi^\ominus(x^0)}|$, $|S_l| \leq \width[x^0]$, and $s_{<l} = |S_1| + \ldots + |S_{l - 1}|$.
\end{proof}

The true power of \lang{RandomAscent} is not just its simplicity, but also its ability to be combined with any other ascent-following algorithm $A$ to make that algorithm efficient on conditionally-smooth landscapes.
Specifically, consider any local search algorithm $A$ that follows ascents.
The $\epsilon$-\lang{shaken}-$A$ algorithm works as follows:
With probability $1 - \epsilon$ it takes a step according to $A$, but with
probability $\epsilon$ it takes a step according to \lang{RandomAscent} (i.e., takes a random fitness increasing step).
This gives an obvious upper-bound on the runtime:

\begin{proposition}\label{prop:runningtime-shaken-ascent}
    On a $\prec$-smooth landscape $f$ on $n$ bits, the expected number of steps taken by $\epsilon$-\lang{shaken}-$A$ to find the peak from initial assignment $x^0$ is: 
\begin{align}
\leq \frac{1}{\epsilon}\Bigg(|\overline{\phi^\ominus(x^0)}| + \text{width}_f(x^0)(1 + \log \text{width}_f(x^0)) \binom{\text{height}_f(x^0) - 1}{2}\Bigg) \\
\leq \frac{1}{\epsilon}\Bigg(n + \text{width}(\prec)(1 + \log \text{width}(\prec))\binom{\text{height}(\prec) - 1}{2}\Bigg).
\end{align}
\label{prop:Shaken_Bound}
\end{proposition}

\subsection{Ascent-biased random local search}\label{app:effectiveFitnessDecreasing}

\noindent\cref{thm:main_stepsBound} also applies to algorithms that do not always go uphill, but occasionally take downhill steps like simulated annealing.

\lang{SimulatedAnnealing} considers a bit to flip uniformly at random, if the considered flip increases fitness then it is accepted but the considered flip decreases fitness then it is only accepted with a probability that decreases at each step~\cite{LocalSearch_Book1}.
More formally, if $Y^\text{SA}_{t + 1}(x) \sim \lang{SimulatedAnnealing}^1_f(x^t)$ then

\begin{equation}
\text{Pr}\{Y_{t  +1} = y\} = \begin{cases}
\frac{1}{n} & \text{ if } i \in \phi^+(x) \text{ and } y = x[i \mapsto \overline{x_i}] \\
\frac{1}{n}\exp(\frac{f(x^t) - f(x^t[i \mapsto \overline{x^t_i}])}{K(t)}) & \text{ if } i \in \phi^-(x) \text{ and } y = x[i \mapsto \overline{x_i}] \\
1 - \frac{|\phi^+(x)|}{n} - Z\frac{|\phi^-(x)|}{n} & \text{ if } y = x
\end{cases}
\label{eq:SA}
\end{equation}

\noindent where $Z = \frac{1}{|\phi^-(x)|}\sum_{i \in \phi^-(x)} \exp(\frac{f(x^t) - f(x^t[i \mapsto \overline{x^t_i}])}{K(t)})$ is the average downstep probability and $K(t)$ is a sequence of temperatures decreasing toward $0$ (known as the \emph{annealing schedule}).
Sometimes the probability of a downward step $r_t(\Delta f) = \exp(\frac{-\Delta f}{K(t)})$ is replaced by other functional forms and other annealing schedules.
However, they always go monotonically to $r_t(\Delta f) \rightarrow 0$ as $t \rightarrow \infty$ for any $\Delta f > 0$ and if $a < b$ then $r(a) \geq r(b)$.

\begin{lemma}
For a $\prec$-smooth fitness landscape on $n$ bits with total order $\prec$, if the probability of accepting a downstep is bounded by $r$ in $\lang{SimulatedAnnealing}^1_f$ then $p(n,l) = n(1 + r)^{n - l}$ satisfies the condition from \cref{eq:drift_time_bound} of \cref{thm:main_stepsBound}.
\label{lem:SA_stepBound}
\end{lemma}

\begin{proof}
Let $P(n,k) = \sum_{j = 0}^k (j + 1)p(n,n-j)$.
Notice that:
\begin{align}
p(n,n-(k+1)) & = 1 +  \underbrace{\frac{n - 1 - r(k + 1)}{n}}_{\text{probability flip doesn't change level}}p(n,n - (k + 1)) \nonumber\\
& \quad \quad + r\underbrace{\frac{P(n,k+1) - p(n,n - (k + 1))}{n}}_{\text{expected time to reach } X_{<l} \text{ given flip in } \phi^{\ominus}(x)} \\
& = n + rP(n,k)
\end{align}
satisfies \cref{eq:drift_time_bound} of \cref{thm:main_stepsBound}.
From this, we get $P(n,k+1) = n + (1 + r)P(n,k)$ and $P(n,0) = n$, which is a recurrence solved by $P(n,k) = \frac{n}{r}((1 + r)^{k + 1} - 1)$.
And thus $p(n,n - k) = n(1 + r)^k$
\end{proof}

As with \lang{RandomAscent} before it, the bound from \cref{lem:SA_stepBound} on the expected number of steps to decrement height can be fed carefully into \cref{thm:main_stepsBound} to bound the expected total number of steps for \lang{SimulatedAnnealing}:

\propSAstepsBound*

\begin{proof}\label{proof:prop:SA_stepsBound}
First wait for the algorithm to take $\tau^\alpha$ steps to reduce $r_t(1) \leq \frac{\alpha}{n}$.
If $f$ is $\prec'$-smooth then extend $\prec'$ to a total order $\prec$, use this with \cref{lem:SA_stepBound} to get an upper bound on the remaining steps:
\begin{align}
\sum_{k = 0}^{n - 1} p(n,n - k) & = \sum_{k = 0}^{n - 1} n(1 + r)^k
= n\frac{1 - (1 + r)^n}{1 - (1 + r)} = n\frac{(1 + r)^n - 1}{r} \\
& = \frac{n^2}{\alpha}((1 + \frac{\alpha}{n})^n - 1) \leq n^2\frac{(\exp(\alpha) - 1)}{\alpha}.
\end{align}
\end{proof}

\subsection{History-based local search}\label{app:efficientHistory}

The algorithms considered so far use randomness to avoid following exponentially long paths to the peak in conditionally-smooth landscapes. 
But randomness is not the only thing that can help local search.
Some local-search algorithms popular in the linear programming literature keep a history of the variables they have altered in their search \cite{historyRules}.
They then use this history to avoid flipping the same subset of variables many times whilst there are still untried variables, which works well in conditionally-smooth landscapes. 
Below we give examples of the variety of history-based rules that can be devised to yield different local search algorithms.

These algorithms all keep a history array $B$ where $B(i)$ stores information about variable $x_i$.
In some cases, the algorithms track information about $x_i=0$ separately form $x_i=1$, doubling the size of the history array with $B(i-)$ tracking $x_i=0$ and $B(i+)$ tracking $x_i=1$.
Distinguishing the history of $x_i=0$ from that of $x_i=1$ will give different algorithmic behavior than tracking simply changes to the variable $x_i$ (see \cite{historyRules} for illustration).
In the examples that follow, we describe a step of the algorithm and de information stored in the array $B$:

\lang{ZadehsRule}~\cite{zadeh} and \lang{LeastUsedDirection}~\cite{historyRules} select the improving bit-flip that has been used the least often thus far. 
For \lang{ZadehsRule} $B(i+)$ stores the number of times $x_i$ has flipped from $x_i=0$ to $x_i=1$, and $B(i-)$ the number of times it has changed from $x_i=1$ to $x_i=0$. \lang{LeastUsedDirection} on the other hand only stores, in $B(i)$, the number of times $x_i$ was flipped.  
The \lang{LeastRecentlyBasic} \cite{cunningham1979} selects the improving bit-flip that was flipped the least-recently.
Here $B(i+)$ stores the iteration number at which $x_i$ was last $1$ and $B(i-)$ at which $x_i$ was last $0$. 
\lang{LeastRecentlyEntered} \cite{leastRecentlyEntered} selects the improving bit-flip such that the new value of the selected bit was held by that bit least-recently, among all bits with an improving flip.
Here $B(i+)$ also stores the iteration number when $x_i$ was last $1$. 
\lang{LeastIterationsInBasis} \cite{historyRules} select the improving bit-flip, such that the new value of the selected bit has been held by that bit the least number of iterations.
Here $B(i+)$ is the number of iterations $x_i$ has been $1$.
Finally, the \lang{LeastRecentlyConsidered}~\cite{cunningham1979} works by fixing an ordering $v_1,\ldots, v_{2n}$ of the $2n$ possible flips. If the last used flip is $v_i$, \lang{LeastRecentlyConsidered} selects the improving flip $s$ that first appears in the sequence $v_{i+1}, v_{i+2}, ..., v_{2n}, v_1, ..., v_{i-1}$, then updated the history array as $B(i) \gets{} ((B(i) - B(s) - 1) \mod{} 2n) + 1)$).

We now prove that, starting at assignment $x^0$, \lang{ZadehsRule} flips the border $\phi^{\oplus}(x^0)$ before any variable $x_i$ can flip three times.
The proof can easily be modified to prove the same for any of the history-based algorithms mentioned above.
\begin{lemma}\label{app:zadehborder}
    Let $x^0,x^1,\ldots x^t$ be the ascent followed by \lang{ZadehsRule} starting from $x^0$.
    If variable $x_i$ flipped three times in this sequence, then any $\phi^\oplus(x^0)\subseteq\phi^\ominus(x^t)$.
\end{lemma}
\begin{proof}
    From $x^0,x^1,\ldots, x^t$ being an ascent, it follows that 
    $\phi^{\oplus}(x^0)\subseteq \phi^{+}(x^{t})\cup \phi^{\ominus}(x^{t})$.
    To show that $\phi^{\oplus}(x^0)\cap\phi^+(x^t)=\emptyset$ we will show that if an arbitrary variable $x_i$ flips three times, then for all $j\in\phi^+(x^t)$ either $B(j+)=1$ or $B(j-)=1$ at the time of the third flip of $x_i$.
    This is sufficient because any $i\in \phi^\oplus(x^0)$ is also in $\phi^+(x^t)$ if and only if it does not flip between time $0$ and $t$ (i.e. at time $t$ we have $B(i+)=B(i-)=0$).    
    Suppose without loss of generality that $t$ is the step at which $x_i$ flipped the third time and that this flip was of the kind $x^t=x^{t-1}[i\mapsto 1]$.
    Then at time $t-1$ we have $B(i+)=1$.
    Since $x_i$ was flipped we know for all variables $k\in\phi^+(x^{t-1})$ at time $t-1$ it was the case that $B(k+)\geq 1$ if $x^{t-1}_k=0$ and $B(k-)\geq 1$ if $x^{t-1}_k=1$.
    Since only $B(i+)$ is updated at time $t$,
    it follows that for all $j\in \phi^\oplus(x^0)\cap\phi^{+}(x^t)$ either $B(j+)=1$ or $B(j-)=1$. 
    In other words $\phi^\oplus(x^0)\subseteq\phi^\ominus(x^t)$.
\end{proof}

\begin{lemma}\label{lem:HB_stepsBound}
        Let $X_{<l}=\{x \; | \; \text{height}(x)<l\}$ and $s_{<l} = \max_{x\in X_{<l}} |\phi^+(x)|$. 
        Then 
    $p(n,l) = 2 s_{<(l+1)}$
    satisfies the condition from \cref{eq:drift_time_bound} of \cref{thm:main_stepsBound} for \lang{ZadehsRule}.
\end{lemma}
\begin{proof}
    This follows from the fact that the \lang{ZadehsRule} only flips variables with indices in $\phi^+(x)$ for any visited $x$ (i.e. only makes improving flips) and from the fact that by \cref{app:zadehborder}, before any variable $x_i$ is flipped a third time, all variables with indices in $S_l\cap \phi^\oplus(x)$ will have flipped, for any starting $x$ with $\text{height}(x)=l$.  
\end{proof}

\begin{proposition}\label{prop:HB_stepsBound}
    On conditionally-smooth fitness landscape $f$ on $n$ bits, the expected number of steps taken by \lang{ZadehsRule} is $\leq 2\width[\prec]\height[\prec]$
\end{proposition}
\begin{proof}
    Follows from feeding the bound of \cref{lem:HB_stepsBound} and $q(n)=\width[\prec]$ into \cref{thm:main_stepsBound}
\end{proof}

It is not hard to adapt \cref{prop:HB_stepsBound} to work for \lang{LeastUsedDirection},  
\lang{LeastRecentlyBasic}, 
\lang{LeastRecentlyEntered},  
\lang{LeastIterationsInBasis}, and \lang{LeastRecentlyConsidered}. 

\subsection{Non-adjacent local search}
\label{app:efficientJumping}

The algorithms that we considered in \cref{app:efficientWalking}, \cref{app:effectiveFitnessDecreasing}, and \cref{app:efficientHistory} transitioned, at each step, from a current assignment to an assignment that is adjacent to it.
We call any algorithm which allows a transition from $x$ to $y$ where $y$ differs from $x$ at more than one variable index a ``jumping algorithm''. 
Jumping algorithms like \lang{AntipodalJump}, \lang{JumpToBest}, and \lang{RandomJump} are especially popular on semismooth fitness landscapes where flipping any combination of variables with indexes in $\phi^+(x)$ results in a fitness increasing jump~\cite{HowardsRule,simplexSurvey,LongJump}.
But jumping algorithms are also used in other contexts like the Kernighan-Lin heuristic for \lang{MaxCut}~\cite{KL_Original}.

Let us start our analysis with the determistic jump rules.
\lang{AntipodalJump} and \lang{JumpToBest} are the easiest to define by
$\lang{AntipodalJump}^1_f(x)  = x[\phi^+(x) \mapsto \overline{x[\phi^+(x)]}]$ and $\lang{JumpToBest}^1_f(x)  = \argmax_{y \in \{0,1\}^{\phi^+(x)}x[\phi^-(x)]} f(y)$.
Notice that $\lang{JumpToBest}^1_f$ is not efficiently computable in general, but we are not interested in the internals of steps but just in the number of steps (i.e, number of applications of $\lang{JumpToBest}^1_f$).
For this, it is important to notice that $\lang{JumpToBest}^1_f$ always flips the border:

\begin{proposition}\label{prop:jump-best-flips-border}
    Define $\phi^\text{JtB}(x)$ such that $\lang{JumpToBest}_f^1(x) = x[\phi^\text{JtB}(x) \mapsto \overline{x[\phi^\text{JtB}(x)]}]$
    then $\phi^\oplus(x)\subseteq \phi^\text{JtB}(x)$.
\end{proposition}

\begin{proof}
    For any assignment $z\in \{0,1\}^{\overline{\phi^\ominus(x)}}x[\phi^\ominus(x)]$ and every $j\in\phi^\oplus(x)$ we have $f(z[j\mapsto \overline{x_j}])>f(z[j\mapsto x_j])$.
    Since $y := \lang{JumpToBest}_f^1(x)$ is a local maximum in this face, we must have $y_j = \overline{x_j}$ and so $\phi^\oplus(x) \subseteq \phi^\text{JtB}(x)$.
\end{proof}

Given $\phi^\oplus(x) \subseteq \phi^\text{JtB}(x) \subseteq \phi^+ \subseteq \overline{\phi^\ominus(x)}$, we conclude that $q(n) = 1$ satisfies \cref{eq:drift_time_bound} from \cref{thm:main_stepsBound} for both \lang{JumpToBest} and \lang{AntipodalJump}, so:

\begin{proposition}\label{prop:runningtime-antipodal-and-best-jump}
    Let $f$ be a conditionally-smooth landscape then \lang{AntipodalJump} and \lang{JumpToBest} take at most $\text{height}_f(x^0)$ steps to find the peak from an initial assignment $x^0$. 
\end{proposition}

The third deterministic jump-rule that we want to analyze is the Kernighan-Lin heuristic~\cite{KL_Original,PLS,PLS_Survey}.
Here, we will also establish that $q(n) = 1$ is a bound, but it will require more effort.
But first, to define the Kernighan-Lin heuristic, we need to define the KL-neighbourhood $\lang{KLN}_f(x)$ in \cref{alg:KLN}.

\begin{algorithm}[htb]
\caption{$\lang{KLN}_f(x)$:}
\begin{algorithmic}[1]
    \Require Fitness function $f$ and assignment $x$.
    \State Initialize $S \gets [n]$ and $y \gets x$ and $\text{KLN} \gets \emptyset$
    \While {$S \neq \emptyset$} 
    \State Let $i \in S$ be lowest index such that $\forall j \in S \quad f(y[i \mapsto \overline{y_i}]) \geq f(y[j \mapsto \overline{y_j}])$
    \State $\text{KLN}(x) \gets \text{KLN} + \{ y[i \mapsto \overline{y_i}] \}$
    \State $y \gets y[i \mapsto \overline{y_i}]$
    \State $S \gets S - \{ i \}$ 
    \EndWhile
    \State \Return $\text{KLN}$
\end{algorithmic}
\label{alg:KLN}
\end{algorithm}

Where for analysis, we will number the neighbours in $\lang{KLN}_f(x)$ as $y^1,...,y^n$ based on the order that they were added to the set and define $\Delta_i = f(y^i) - f(y^{i - 1})$ (where $y^0 = x$).
Now, we can define $\lang{KerninghanLin}^1_f(x) = \argmax_{y \in \lang{KLN}_f(x)} f(y)$.
Unlike $\lang{JumpToBest}^1_f$, the $\argmax$ in $\lang{KerninghanLin}^1_f$ can always be computed efficiently.
This is one of the reasons that the \lang{KernighanLin} was presented as a heuristic for solving \lang{MaxCut}~\cite{KL_Original} -- which is \PLS-complete in the weighted version.
It was also idealized and widely studied in the literature on polynomial local search, including in the more general case of \wSAT~\cite{PLS,PLS_Survey}.

If, as with \lang{JumpToBest}, we define $\phi^\text{KL}(x)$ such that $\lang{KernighamLin}^1_f(x)$ returns $x[\phi^\text{KL}(x) \mapsto \overline{x[\phi^\text{KL}(x)]}$ then it has the following properties:

\begin{proposition}\label{prop:KLH-flips-sources}
Given a fitness landscape $f$, assignment $x$ and $y^i \in \lang{KLN}_f(x)$,
define $S^i$ such that $y^i = x[S^i \mapsto \overline{x[S^i]}]$, let $k$ be the first index such that $\Delta_k > 0$ but $\Delta_{k + 1} \leq 0$.
If $f$ is conditionally-smooth then $\phi^\oplus(x) \subseteq S^k \subseteq \phi^\text{KL}(x)$.
If $f$ is also semismooth then $S^k = \phi^\text{KL}(x) \subseteq \overline{\phi^\ominus(x)}$.
\end{proposition}

\begin{proof}
    For any assignment $z\in \{0,1\}^{\overline{\phi^\ominus(x)}}x[\phi^\ominus(x)]$ and every $j\in\phi^\oplus(x)$ we have $f(z[j\mapsto \overline{x_j}])>f(z[j\mapsto x_j])$.
    Since $\Delta_t > 0$ for all $t < k + 1$, $y^1,\ldots,y^k$ is an ascent.
    Since $\Delta_{k + 1} \leq 0$, it follows that $y^k$ is a local peak in $\{0,1\}^{\phi^-(x)}x[\phi^-(x)]$, so $y^k_j=\overline{x_j}$ and thus $\phi^\oplus(x)\subseteq S^k \subseteq \overline{\phi^\ominus(x)}$.
    
    If $f$ is semismooth then since $y^k$ is a local maximum in $\{0,1\}^{\phi^-(x)}x[\phi^-(x)]$ then it is also a global maximum in that face.
    Thus $\phi^\text{KL}(x) = S^k$.
\end{proof}

\cref{prop:KLH-flips-sources} means that for landscapes that are both conditionally-smooth and semismooth -- like oriented \VCSP s, for example --  setting $q(n) = 1$ satisfies \cref{eq:drift_time_bound} from \cref{thm:main_stepsBound} for \lang{KerninghanLin} and so it follows that:

\begin{proposition}
If $f$ is a conditionally-smooth and semismooth fitness landscape then \lang{KerninghanLin} takes at most $\text{height}_f(x^0)$ steps to find the peak from initial assignment $x^0$.
\label{prop:KL_bound}
\end{proposition}

Finally, for a slightly more complex form for $q(n)$ we can consider a stochastic jump rule like \lang{RandomJump} that given an assignment $x$ flips a random subset of $\phi^+(x)$.
Or, slightly more formally, if $Y^\text{RJ}(x) \sim \lang{RandomJump}^1_f(x)$ then:

\begin{equation}
\text{Pr}\{Y^\text{RJ}(x) = x[S \mapsto \overline{x[S]}]\} = \begin{cases}
\frac{1}{2^{|\phi^+(x)|}} & \text{ if } S \subseteq \phi^+(x) \\
0 & \text{ otherwise }
\end{cases}
\label{eq:randomJump}
\end{equation}

\noindent This stochastic local search algorithm requires only a log-factor more steps than the deterministic jump rules that we considered:

\propRJstepsBound*
\begin{proof}\label{proof:prop:randomJump_steps}
    Consider an any assignment $x$ with $l = \height[x]$.
    Let $x[S\mapsto\overline{x[S]}] \sim Y^\text{RJ}(x)$.
    For every $i\in S_l$ the probability that $i\in S$ is $1/2$, so $\mathbb{E}(|S_l\cap S)|)=|S_l|/2$.
    Thus $p(n,l) = \log |S_l| + 2$ and $q(n) = \log(\width[x^0]) + 2$ will satisfy \cref{eq:drift_time_bound} in \cref{thm:main_stepsBound}.
\end{proof}

\newpage
\bibliography{AUSO_VCSP}

\end{document}